\documentclass[twocolumn,aps,nofootinbib,superscriptaddress,tightenlines,notitlepage,pra]{revtex4-2}

\usepackage{amsmath, mathtools, bbm}
\usepackage{amsthm}
\usepackage{amsfonts}
\usepackage{color}
\usepackage{graphicx}
\usepackage{soul}

\usepackage[ breaklinks,
			 colorlinks = true,
             linkcolor = blue,
             urlcolor  = blue,
             citecolor = red,
             anchorcolor = blue,
]{hyperref}

\newtheorem{lemma}{Lemma}

\newtheorem{theorem}{Theorem}

\newtheorem{prop}[lemma]{Proposition}

\usepackage[capitalise]{cleveref}

\DeclareMathOperator{\tr}{tr}
\DeclareMathOperator{\erf}{erf}
\DeclareMathOperator{\sgn}{sgn}
\newcommand{\id}{1}

\newcommand{\eps}{\varepsilon}

\newcommand{\ket}[1]{|#1\rangle}
\newcommand{\bra}[1]{\langle#1|}
\newcommand{\E}{\mathbb{E}}
\renewcommand{\Pr}{\mathbb{P}}

\newcommand{\suppress}[1]{}
\DeclareMathOperator*{\argmin}{arg\,min}

\newcommand{\ceil}[1]{\lceil #1\rceil}
\newcommand{\floor}[1]{\lfloor #1\rfloor}

\usepackage{algorithmicx}
\usepackage{algorithm}
\usepackage[noend]{algpseudocode}
\usepackage{footnote}

\usepackage[shortlabels]{enumitem}

\usepackage{float}

\makeatletter 
    
\renewcommand\onecolumngrid{
\do@columngrid{one}{\@ne}%
\def\set@footnotewidth{\onecolumngrid}
\def\footnoterule{\kern-6pt\hrule width 1.5in\kern6pt}%
}

\begin{document}

\title{\Large A randomized quantum algorithm for statistical phase estimation}

\author{Kianna Wan}
\affiliation{AWS Center for Quantum Computing, Pasadena, USA}
\affiliation{Stanford Institute for Theoretical Physics, Stanford University, Stanford, USA}

\author{Mario Berta}
\affiliation{AWS Center for Quantum Computing, Pasadena, USA}
\affiliation{Institute for Quantum Information and Matter, California Institute of Technology, Pasadena, USA}
\affiliation{Department of Computing, Imperial College London, London, UK}

\author{Earl T.~Campbell}
\affiliation{AWS Center for Quantum Computing, Cambridge, UK}


\begin{abstract}

Phase estimation is a quantum algorithm for measuring the eigenvalues of a Hamiltonian. We propose and rigorously analyse a randomized phase estimation algorithm with two distinctive features. First, our algorithm has complexity independent of the number of terms $L$ in the Hamiltonian. Second, unlike previous $L$-independent approaches, such as those based on qDRIFT, all sources of error in our algorithm can be suppressed by collecting more data samples, without increasing the circuit depth.

\end{abstract}

\maketitle

\section{Introduction}

Quantum computers can be used to simulate dynamics and learn the spectra of quantum systems, such as interacting particles comprising complex molecules or materials, described by some Hamiltonian $H$. Phase estimation~\cite{kitaev2002classical} on the unitary $U=e^{i H t}$ efficiently solves the common spectral problem of computing ground state energies, whenever we can efficiently prepare a trial state with non-trivial (not exponentially small) overlap $\eta$ with the ground state~\cite{Abrams1999} (see also~\cite{Poulin2018}). Each run of standard phase estimation returns a single eigenvalue, with precision and success probability dependent on the number of times $U$ is used. 

Recently, statistical approaches to phase estimation have been proposed~\cite{obrien19,Terhal21,lin21}, where each run uses only a few ancillae and shorter circuits than standard phase estimation.  As such, statistical phase estimation may be better suited to early fault-tolerant quantum computers that are qubit- and depth-limited. However, in these approaches, a single run  gives a sample of an estimator for $\langle U^{j} \rangle$ for some runtime $j$, which alone is not enough to infer spectral properties. Multiple runs with different values of $j$ are needed, and statistical analysis gives spectral information with a confidence that increases with the amount of data obtained. These runs could be massively parallelized across multiple quantum computers. Interestingly, the approach of Lin \& Tong~\cite{lin21} is not only statistical in its analysis, but also generates the runtimes $j$, and therefore the circuits, from a random ensemble.

The cost of phase estimation---statistical or standard---typically depends on the Hamiltonian sparsity $L$, the number of terms in the Hamiltonian when decomposed in a suitable basis, such as the Pauli basis. Simple schemes based on implementing $U$ using Trotter formulae have $\mathcal O(L)$ gate complexity \cite{poulin2014trotter,babbush2015chemical,kivlichan2020improved,campbell2020early,mcardle2021exploiting}.
This can be prohibitive for the electronic structure problem in chemistry and materials science, where typically $L=\mathcal O(N^4)$ for an $N$-orbital problem \cite{helgaker2014molecular}. This increases to $L=\mathcal O(N^6)$ when using transcorrelated orbitals~\cite{motta2020quantum,mcardle2020improving} to better resolve electron-electron interactions. Interestingly, sub-linear non-Clifford complexity $\mathcal O(\sqrt{L}+N)$ is possible by employing an efficient data-lookup oracle~\cite{Babbush2018,low2018trading} in qubitization-based implementations of phase estimation~\cite{babbush2019quantum,berry2019qubitization,vonburg2021DoubleFactorized,Lee20}. However, these approaches require $\mathcal O(\sqrt{L})$ ancillae, which increases the qubit cost from $\mathcal O(N)$ to $\mathcal O(N^2)$, or even $\mathcal O(N^3)$ in the transcorrelated setting.


Heuristic truncation and low-rank factorisations have been proposed to decrease the sparsity $L$~\cite{berry2019qubitization,Lee20,vonburg2021DoubleFactorized} of the electronic structure Hamiltonian. As an alternative approach, randomized compilation~\cite{campbell19,Wiebe19,ouyang2020compilation} has been rigorously shown to enable phase estimation with gate complexity that is independent of $L$ for any Hamiltonian. A weakness of these randomized algorithms is a systematic error in energy estimates that can only be suppressed by increasing gate complexity, leading to high gate counts per run (cf.~\cite[Appendix D]{Lee20}).

Here, we overcome this difficulty by combining the statistical approach of Lin \& Tong \cite{lin21} with a novel random compilation of each $U^j$ instance, that has parallels to---but is distinct from---both the qDRIFT random compiler~\cite{campbell19} and the linear combinations of unitaries (LCU) method~\cite{childs2012,berry15}. Our algorithm for phase estimation is doubly randomized in that we randomly sample $j$, then approximate $U^j$ using a random gate sequence.  Unlike in 
any previous approach, all approximation and compilation errors can be expressed in terms of statistical noise that is suppressed by collecting more data samples. This allows for a trade-off between the gate complexity per sample and the number of samples required.  We explore this trade-off and show how to efficiently find the algorithmic parameters that minimise the total complexity. In contrast, qDRIFT approximates $U$ up to some systematic error (measured by the diamond norm) that cannot be mitigated by increasing the number of samples.

Applied to ground state energy estimation, we can tune the gate vs.\ sample trade-off to yield the following complexities. Given a Hamiltonian as a linear combination of Pauli operators with total weight $\lambda$, and an ansatz state with overlap at least $\eta$ with the ground space, we can choose to sample from $\widetilde{\mathcal{O}}(\eta^{-2})$ randomly compiled quantum circuits, where $\widetilde{\mathcal{O}}(\cdot)$ hides polylogarithmic factors. Each circuit uses one ancilla and at most $\widetilde{\mathcal{O}}(\lambda^2\Delta^{-2})$ single-qubit Pauli rotations to estimate the ground state energy to within additive error $\Delta$.

In Section~\ref{sec:threshold}, we start by constructing a subroutine that we refer to as \textit{eigenvalue thresholding}, which we then apply to ground state energy estimation in Sec.~\ref{sec:minimal}. We discuss examples from quantum chemistry in Sec.~\ref{sec:examples}.



\section{Eigenvalue thresholding}\label{sec:threshold}

\paragraph*{Problem setting.} We assume that the Hamiltonian $H$ is specified as a linear combination of $n$-qubit Pauli operators $P_{\ell}$:
\begin{align}\label{eq:hamiltonian}
H=\sum_{\ell=1}^L\alpha_\ell P_\ell, \quad \text{with } \lambda \coloneqq \sum_{\ell=1}^L |\alpha_\ell|.
\end{align}
This form can always be achieved, and is particularly natural for many physical systems of interest, such as fermionic Hamiltonians~\cite{jordan28,Verstraete2005,Seeley2012,Havlek2017,Derby2021}. Note that the spectral norm $\|H\|$ obeys the generally loose bound $\|H\|\leq\lambda$. We consider the following problem of coarsely determining whether an ansatz state $\rho$ has overlap with eigenstates of $H$ with eigenvalues below some threshold: Given a threshold $X$, precision $\Delta >0$, and overlap parameter $\eta > 0$, we seek to decide if (A) $\tr[\rho \Pi_{\leq X - \Delta}] < \eta$ or (B) $\tr[\rho \Pi_{\leq X + \Delta}] > 0$, where $\Pi_{\leq x}$ denotes the projector onto the eigenspaces of $H$ with eigenvalues at most $x$. Both of these statements can simultaneously be true, in which case it suffices to output either A or B. We refer to this problem as \textit{eigenvalue thresholding}, and its solution will later allow us to estimate the ground state energy, given a suitable ansatz $\rho$. 


\paragraph*{Cumulative distribution function.}

Similarly to~\cite{lin21}, we define the cumulative distribution function (CDF) associated with the Hamiltonian $H$ and ansatz state $\rho$ as 
\begin{equation} \label{cdfdef}
C(x) \coloneqq \tr\big[\rho \Pi_{\leq x/\tau}\big],
\end{equation}
where $\tau \coloneqq \frac{\pi}{2\lambda + \Delta}$ is a normalisation factor. The jump discontinuities in $C(x)$ occur at eigenvalues of $\tau H$, so appropriately characterising the CDF would enable us to estimate the spectrum of the Hamiltonian. 
We can write $C(x)$ as the convolution $(\Theta * p)(x)$ of the Heaviside function $\Theta(\cdot)$ and the probability density function $p(\cdot)$ corresponding to $\tau H$ and $\rho$:
\begin{equation} \label{cdfconv}
    C(x) 
    =\int_{-\pi/2}^{\pi/2}dy\,p(y)\Theta(x-y),
\end{equation}
noting that $p(x)$ is supported within $x \in (-\frac{\pi}{2},\frac{\pi}{2})$ since $\tau \|H\| \leq \tau \lambda < \frac{\pi}{2}$.\footnote{This will enable us to replace $\Theta(x)$ with a periodic function that is a good approximation only within $x \in (-\pi,\pi)$.}
Eigenvalue thresholding then reduces to the following problem regarding the CDF.

\smallskip
\noindent \textbf{Problem 1:} For given $x \in [-\tau\lambda, \tau\lambda]$ and $\delta > 0$, determine whether
\begin{equation} \label{cdfprob}
C(x - \delta) < \eta \quad \text{or} \quad C(x + \delta) > 0\,,
\end{equation}
outputting either statement if both are true.
\smallskip

\noindent In particular, solving Problem~1 for $x = \tau {X}$ and $\delta = \tau\Delta$ solves eigenvalue thresholding.\footnote{Solving Problem~1 with these parameter values also solves the ``eigenvalue threshold problem''~\cite{martyn21,gilyen19}, which, unlike eigenvalue thresholding, is a promise problem (where it is guaranteed that either $\tr[\rho\Pi_{\leq X -\Delta}] \geq \eta$ or $\tr[\rho \Pi_{\leq X + \Delta}] = 0$) and hence cannot be used as a subroutine for phase estimation in the same manner.} 

\paragraph*{Algorithm overview.} To solve Problem~1, we will construct an approximation $\widetilde{C}(\cdot)$ to the CDF $C(\cdot)$ satisfying
\begin{equation} \label{acdf}
    C(x-\delta) - \varepsilon \leq \widetilde{C}(x) \leq C(x +\delta) + \varepsilon
\end{equation}
for relevant values of $x$, $\delta$, and $\varepsilon$. Observe that for $\varepsilon \in (0,\eta/2)$, $\widetilde{C}(x) < \eta - \varepsilon$ would imply the first case of Eq.~\eqref{cdfprob}, while $\widetilde{C}(x) > \varepsilon$ would imply the second case. Hence, it suffices to estimate $\widetilde{C}(x)$.

Our algorithm is based on expressing $\widetilde{C}(x)$ in terms of a linear combination of computationally simple unitaries,  obtained via a two-step construction. First, we develop an improved Fourier series approximation to the Heaviside function (Lemma~\ref{lem:fourier}).
Second, we combine this with a novel decomposition of the time evolution operators (Lemma~\ref{lem:simulation}) in the relevant Fourier series. Randomly sampling unitaries from our decomposition and estimating their expectation values using Hadamard tests (Fig.~\ref{fig:circuits}(a)) will give estimates for $\widetilde{C}(x)$, allowing us to solve Problem~1 with high probability. 

\begin{figure*}
    \centering
    \includegraphics[width=0.99\textwidth]{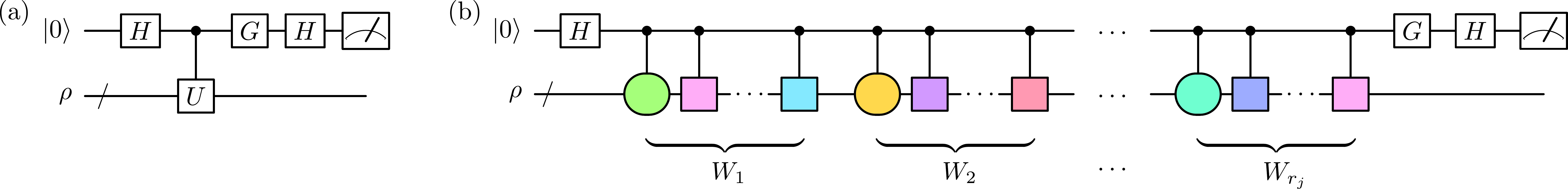}
    \caption{(a) Hadamard test on $\rho$ and $U$: setting $G = \mathbbm{1}$ (resp.\ $G= S^\dagger \coloneqq \ket{0}\bra{0} - i\ket{1}\bra{1}$) and associating the measurement outcomes $(\ket{0},\ket{1})$ with $(+1,-1)$ produces an unbiased estimator for $\mathrm{Re}(\tr[\rho U])$ (resp.\ $\mathrm{Im}(\tr[\rho U])$).
    (b)~Schematic depiction of the randomly compiled circuits in our algorithm. For $\hat{H} = \sum_\ell p_\ell P_\ell$, the squares represent Pauli operators randomly drawn from $\{P_\ell\}_\ell$ according to $\{p_\ell\}_\ell$, while circles denote multi-qubit Pauli rotations; see the proof of Lemma~\ref{lem:simulation} for details. The number of Pauli operators appearing in each $W_i$ is random, and will be zero with high probability.}
\label{fig:circuits}
\end{figure*}


\paragraph*{Fourier series approximation.}

Following {Lin \& Tong~\cite{lin21}, which uses ideas similar to those in~\cite{Somma2002,Somma2019}, we obtain an approximate CDF $\widetilde{C}(\cdot)$ by replacing $\Theta(\cdot)$ in Eq.~\eqref{cdfconv} with a finite Fourier series approximation thereof. As in \cite{vanapeldoorn20,gilyen19} and related works, we need a Fourier series with small approximation error on $|x| \in [\delta, \pi -\delta]$ for fixed $\delta>0$, small total weight of Fourier coefficients, and small maximal ``time'' parameter $|t|$ in the $e^{itx}$ terms. We explicitly construct such a Fourier series in Appendix~\ref{app:fourier}.

\begin{lemma}\label{lem:fourier}
For any $\varepsilon>0$ and $\delta \in (0,\pi/2)$, the Fourier series $F(x) = \sum_{j \in S_1} F_j e^{ijx}$ defined in Eq.~\eqref{Fbetad} with $S_1 \coloneqq \{0\} \cup \{ \pm (2j+1)\}_{j=0}^d$ and $d = \mathcal O(\delta^{-1}\log(\varepsilon^{-1}))$ satisfies
\begin{enumerate}
    \item $|F(x) - \Theta(x)| \leq \varepsilon \quad \forall\, x \in [-\pi + \delta, -\delta] \cup [\delta, \pi-\delta]$, 
    \item $-\varepsilon \leq |F(x)| \leq 1 + \varepsilon \quad \forall\, x \in \mathbb{R}$,
    \item $\mathcal{F}\coloneqq\sum\limits_{j\in S_1} |F_j| =  \mathcal O(\log d)$.
\end{enumerate}
\end{lemma}

This improves on the Fourier approximation of Lin \& Tong, which has $d = \mathcal{O}(\delta^{-1}\log(\delta^{-1}\varepsilon^{-1}))$~\cite[Lemma~6]{lin21}. As such, Lemma~\ref{lem:fourier} also improves the asymptotic complexity of their phase estimation algorithm. In Appendix~\ref{app:fourier}, we prove a stronger version of Lemma~\ref{lem:fourier} with explicit constants, by converting suitable Chebyshev approximations to the error function into Fourier series. 

Using the Fourier series $F(\cdot)$ of Lemma~\ref{lem:fourier}, we obtain the approximate CDF 
\begin{equation} \label{Ctilde}
    \widetilde{C}(x) \coloneqq \int_{-\pi/2}^{\pi/2}dy\,p(y)F(x-y) = \sum_{j \in S_1} F_j e^{ijx}\tr[\rho e^{i\hat{H}t_j}],
\end{equation}
where $t_j \coloneqq -j\tau\lambda$ 
and $\hat{H} \coloneqq H/\lambda$. We show in Appendix~\ref{app:acdf} that for any $\delta \in (0,\tau\Delta]$, this $\widetilde{C}(\cdot)$ indeed satisfies the guarantees in Eq.~\eqref{acdf}. 


\paragraph*{LCU decomposition of time evolution operators.}

Instead of directly implementing the time evolution operators $e^{i\hat{H}t_j}$ from Eq.~\eqref{Ctilde} in Hadamard tests, as considered by~\cite{lin21}, we further decompose each of these terms into a specific linear combination of unitaries.

\begin{lemma} \label{lem:simulation}
Let $\hat{H} = \sum_\ell p_\ell P_\ell$ be a Hermitian operator that is specified as a convex combination of Pauli operators. For any $t \in \mathbb{R}$ and $r \in \mathbb{N} \coloneqq \{1,2,\dots\}$, there exists a linear decomposition
\[ e^{i\hat{H}t} = \sum_{k\in S_2} b_k U_k \]
for some index set $S_2$, real numbers $b_k > 0$, and unitaries $U_k$, such that
\[ \sum_{k \in S_2} b_k \leq \exp(t^2/r), \]
and for all $k \in S_2$, the non-Clifford cost of controlled-$U_k$ is that of $r$ controlled single-qubit Pauli rotations.
\end{lemma}

This decomposition is conceptually different from previous LCU methods, cf.~\cite{berry15} and references therein. The purpose of Lemma~\ref{lem:simulation} is to allow for a trade-off between the sample complexity and gate complexity of our algorithm. Specifically, as shown later,
the sample complexity depends on the total weight $\sum_{k \in S_2} b_k$ of the coefficients in our decomposition. Since this is bounded by $\exp(t^2/r)$, we can reduce the sample complexity by increasing $r$, at the cost of increasing the {gate complexity per sample}, and vice versa.

To prove Lemma~\ref{lem:simulation}, we write $e^{i\hat{H}t} = (e^{i\hat{H}t/r})^r$ and Taylor-expand each $e^{i\hat{H}t/r} = \mathbbm{1} + i\hat{H}t/r + \mathcal{O}((t/r)^2)$.  We then pair up consecutive terms in this expansion, which differ in phase by $i$. Since $\hat{H}$ is a convex combination of Pauli operators, this gives rise to convex combinations of multi-qubit Pauli rotations, e.g., the leading term is
\begin{equation} \label{eq:LeadingOrder}
\mathbbm{1} + i\hat{H}t/r = \sum_{\ell}p_\ell ({\mathbbm{1} + iP_\ell t/r}) \propto \sum_\ell p_\ell e^{i\theta P_\ell} 
\end{equation}
with $\theta \coloneqq \arccos\sqrt{1+(t/r)^2}$. The higher-order terms contain additional Pauli operators, as illustrated in Fig.~\ref{fig:circuits}(b). The controlled version of each Pauli rotation can be implemented using a controlled single-qubit rotation, along with Clifford gates. Hence, each controlled-$U_k$ requires $r$ controlled single-qubit rotations in total. Explicit forms for the higher-order terms and proof details are given in Appendix~\ref{app:lcu}, where we also show, via Algorithm~\ref{algsample}, that one can efficiently sample $U_k$ according to the distribution given by $\{b_k\}_k$.  


\paragraph*{Our algorithm for Problem~1.}
Putting together the above results, we apply Lemma~\ref{lem:simulation} to decompose each $e^{i\hat{H}t_j}$ in Eq.~\eqref{Ctilde} as $e^{i\hat{H}t_j} = \sum_{k \in S_2} b_k^{(j)} U_k^{(j)}$. We choose a positive integer $r_j$ for each $j \in S_1$, and define the corresponding ``runtime vector'' $\vec{r} = (r_j)_j \in \mathbb{N}^{S_1}$. This leads to the final decomposition 
\begin{align}\label{eq:sampling}
\widetilde C(x)&=\sum_{(j,k)\in S_1\times S_2}\underbrace{F_{j}e^{ijx}b_{k}^{(j)}}_{=:\;a_{jk}}\tr[\rho U_{k}^{(j)}]
\end{align}
with total weight
\begin{align}
\mathcal{A}(\vec{r})&\coloneqq\sum_{(j,k)\in S_1\times S_2}|a_{jk}|\leq\sum_{j\in S_1}|F_{j}|\exp(t_j^2/r_{j})\,.\label{eq:A-sum}
\end{align}
As a simple example,
\begin{equation} \label{eq:constant-weight}
r_j = \lceil 2t_j^2\rceil \,\enspace \forall\, j \in S_1 \quad \text{gives} \quad \mathcal{A}(\vec{r}) \leq \sqrt{e}\mathcal{F}.
\end{equation}

Recall that we can solve Problem~1 by determining if $\widetilde{C}(x) < \eta - \varepsilon$ or $\widetilde{C}(x) > \varepsilon$. To estimate $\widetilde{C}(x)$, we sample $(j,k)$ from $S_1 \times S_2$ with probability proportional to $|a_{jk}|$, and perform a Hadamard test 
on $\rho$ and $U_k^{(j)}$, obtaining an estimate $m_{jk}$ for $\tr[\rho U_k^{(j)}]$. Then, $z_{jk} \coloneqq \mathcal{A}(\vec{r})e^{i\,\mathrm{arg}(a_{jk})} m_{jk}$ is an unbiased estimate of $\widetilde{C}(x)$. Letting $\overline{Z}$ denote the random variable obtained by taking the average of $\mathcal{C}_{\mathrm{sample}}$ such estimates, it follows from Hoeffding's inequality that
guessing $\widetilde{C}(x) < \eta-\varepsilon$ if $\mathrm{Re}[\overline Z] < \eta/2$, and $\widetilde{C}(x) > \varepsilon$ otherwise, gives a correct answer with probability at least $1-\vartheta$ provided that $\mathcal{C}_{\mathrm{sample}} \geq 4\mathcal{A}(\vec{r})^2(\eta/2-\varepsilon)^{-2}\ln(\vartheta^{-1})$ (cf.~Appendix~\ref{app:prob}).
Thus, we arrive at Algorithm~\ref{alg1}, our algorithm for solving Problem 1, and hence eigenvalue thresholding.

\begin{algorithm}[H]
\caption{algorithm for Problem~1\label{alg1}} 
\textbf{Problem inputs:} an $n$-qubit Hamiltonian $H = \sum_{\ell=1}^L \alpha_\ell P_\ell$ with $\alpha_\ell >0$ and $\lambda \coloneqq \sum_{\ell = 1}^L \alpha_\ell$, an ansatz state $\rho$, a precision parameter $\Delta > 0$; $\tau \coloneqq \frac{\pi}{2\lambda + \Delta}$.

\textbf{Algorithm parameters:}  real numbers $x \in [-\tau\lambda,\tau\lambda]$, $\delta \in (0,\tau\Delta]$, $\eta \in (0,1]$, and $\varepsilon \in (0,\eta/2)$,
a probability $\vartheta$. \\

\textbf{Output:} ${0}$ if $C(x -  \delta) < \eta$, ${1}$ if $C(x + \delta) > 0$, and either $0$ or $1$ if both are true (where $C(\cdot)$ is the CDF defined in Eq.~\eqref{cdfdef}) with probability of error at most $\vartheta$. \\ 
\begin{algorithmic}[1]
\State Compute the coefficients $\{F_j\}_{j \in S_1}$, specified in Eq.~\eqref{eq:Fouriercoefficients}, of the Fourier series from Lemma~\ref{lem:fourier} with approximation parameters $\delta$ and $\varepsilon$. Set $t_j \leftarrow -j \lambda \tau$ $\forall \, j \in S_1$.
\State Choose a runtime vector $\vec{r} \in \mathbb{N}^{S_1}$ (using e.g., Eq.~\eqref{eq:constant-weight} or Eq.~\eqref{eq:r-vector}), and apply Lemma~\ref{lem:simulation} 
to obtain the decomposition in Eq.~\eqref{eq:sampling}, with 
total weight $\mathcal A(\vec r)$ as in Eq.~\eqref{eq:A-sum}. 
\State \begin{equation} \label{Csample} \mathcal{C}_{\mathrm{sample}}(\vec{r}) \leftarrow \left\lceil \left(\frac{2\mathcal{A}(\vec{r})}{\eta/2 - \varepsilon}\right)^2 \ln\frac{1}{\vartheta}\right\rceil. \end{equation}
\State \textbf{For} $i = 1,\dots, \mathcal{C}_{\mathrm{sample}}(\vec{r})$:
\State \hspace{1em} Sample a unitary $U_k^{(j)}$ from Eq.~\eqref{eq:sampling} as follows:
\begin{enumerate}[a.] 
\item \hspace{1em} Sample an index $j \in S_1$ with probability $\propto |F_j|$.
\item \hspace{1em} Sample a unitary using
Algorithm~\ref{algsample} with inputs

\hspace{1em} ${H}/\lambda$, $t_j$, $r_j$ .
\end{enumerate}
\State \label{step6} \hspace{1em} Perform a Hadamard test 
with inputs $\rho$ and $U_k^{(j)}$,

obtaining an estimate $m_i$ of $\tr[\rho U_k^{(j)}]$.
\State \hspace{1em} $z_i \leftarrow \mathcal{A}(\vec{r}) e^{i(\mathrm{arg}(F_j)+jx)} m_i$
\State $\overline{z} \leftarrow \sum_i z_i/\mathcal{C}_{\mathrm{sample}}(\vec{r})$. If $\mathrm{Re}(\overline{z}) < \eta/2$, \textbf{return} $0$. Else, \textbf{return} $1$
\end{algorithmic}
\end{algorithm}


\paragraph*{Complexity.}

The Hadamard test in Step~\ref{step6} is the only quantum step and involves two circuits on $n + 1$ qubits, for an $n$-qubit Hamiltonian $H$. The expected number of controlled Pauli rotations per circuit is
\begin{equation} \label{Cgate}
    \mathcal{C}_{\mathrm{gate}}(\vec{r}) \coloneqq \frac{1}{\mathcal{A}(\vec{r})} \sum_{(j,k) \in S_1 \times S_2}|a_{jk}|r_j.
\end{equation}
Step~\ref{step6} is repeated $\mathcal{C}_{\mathrm{sample}}(\vec{r})$ times, so the expected total non-Clifford complexity is 
$2\mathcal{C}_{\mathrm{sample}}(\vec{r})\cdot \mathcal{C}_{\mathrm{gate}}(\vec{r})$.

It remains to specify how to choose the runtime vector $\vec{r} \in \mathbb{N}^{S_1}$. For example, we could aim to minimise the total complexity
\begin{equation} \label{eq:r-vector}
\argmin_{\vec{r}}\;\mathcal{C}_{\mathrm{sample}}(\vec{r})\cdot \mathcal{C}_{\mathrm{gate}}(\vec{r}).
\end{equation}
\textit{Prima facie} this is a high-dimensional optimisation problem, as $|S_1| = \mathcal{O}(\delta^{-1}\log(\epsilon^{-1}))$ from Lemma~\ref{lem:fourier}. However, differentiating with respect to $\vec{r}$, one sees that the argmin is effectively described by a single free parameter.  Therefore, optimising $\vec{r}$ is reducible to an efficiently solvable one-dimensional problem, and this further holds when minimising $\mathcal{C}_{\mathrm{sample}}$ subject to constraints on $\mathcal{C}_{\mathrm{gate}}$; see Appendix~\ref{app:optimiser} for details. Moreover, if one is exclusively interested in asymptotic complexities, the simple choice for $\vec{r}$ in Eq.~\eqref{eq:constant-weight} already gives 
\begin{equation}
    \mathcal{C}_{\mathrm{sample}}(\vec{r}) = \mathcal{O}\left(\frac{1}{\eta^2}\log^2\Big(\frac{1}{\delta}\log \frac{1}{\eta} \Big)\log \frac{1}{\vartheta} \right) = \widetilde{\mathcal{O}}\left(\frac{1}{\eta^2}\right) \\
\end{equation}
\begin{equation}\label{Cgateasymptotic}
\text{and}\quad\mathcal{C}_{\mathrm{gate}}(\vec{r}) = \mathcal{O}\left(\frac{1}{\delta^2}\log^2\frac{1}{\eta}\right)= \widetilde{\mathcal{O}}\left(\frac{1}{\delta^2}\right),
\end{equation}
since $\mathcal{A}(\vec{r}) \leq \sqrt{e}\mathcal{F}$ and $\mathcal{C}_{\mathrm{gate}}(\vec{r}) \leq \max_{j \in S_1} r_j = 2[(2d+1)\tau\lambda]^2$ for this choice, with $\mathcal{F}$ and $d$ given by Lemma~\ref{lem:fourier} and picking $\varepsilon = \text{const.} \times \eta$ in Algorithm~\ref{alg1}. Note that the \emph{worst-case} gate complexity thus has the same scaling as that in Eq.~\eqref{Cgateasymptotic} for the \emph{expected} gate complexity $\mathcal{C}_{\mathrm{gate}}(\vec{r})$. Hence, we arrive at a total complexity $\widetilde{\mathcal{O}}(\delta^{-2} \eta^{-2})$.
For eigenvalue thresholding, one would choose $\delta = \tau\Delta$, in which case $\delta^{-1} = \mathcal{O}(\lambda/\Delta)$.


\section{Ground state energy estimation}\label{sec:minimal}


Under appropriate assumptions on the Hamiltonian $H$ and ansatz state $\rho$, our method for estimating the CDF
can be adapted to perform phase estimation. Specifically, eigenvalues of $H$ coincide with the locations of jump discontinuities in $C(x)$, and we can estimate these locations given sufficient knowledge about the overlap of $\rho$ with relevant eigenspaces. For simplicity, we restrict ourselves to the problem of estimating the ground state energy $[H]_{\min}$, which only requires the standard promise that $\tr[\rho \Pi_{\min}] \geq \eta$ for some $\eta > 0$, where $\Pi_{\min}$ denotes the projector onto the ground space of $H$.


The analysis in \cite[Section 5]{lin21} shows that by solving Problem 1 for $s = \mathcal{O}(\log(\delta^{-1}))$ different values of $x$ determined in a fashion similar to binary search, one can find an $x^*$ such that $C(x^* - \delta) <\eta$ and $C(x^* + \delta) > 0$, which implies that $|x^* - \tau [H]_{\min}| \leq \delta$. Hence, if we take $\delta = \tau\Delta$, then $x^*/\tau$ would give an estimate of the ground state energy to within additive error $\Delta$. We use Algorithm~\ref{alg1} to solve Problem~1, noting that we can reuse the samples collected in Step~\ref{step6} for \emph{all} of the different $x$ values, with only a small overhead in the sample complexity. Namely, since Algorithm~\ref{alg1} errs with probability at most $\vartheta$ for any $x$, choosing $\vartheta = \xi/s$ would ensure, by the union bound, that the ground state is successfully estimated with probability at least $1-\xi$.


\begin{theorem} \label{thm:main}
For any $n$-qubit Hamiltonian $H$ of the form in Eq.~\eqref{eq:hamiltonian}, let $\rho$ be a state that has overlap $\tr[\rho \Pi_{\min}] \geq \eta$ with the ground space of $H$. Then, the ground state energy of $H$ can be estimated to within additive error $\Delta$ with probability at least $1-\xi$ using $\mathcal{O}\left(\frac{1}{\eta^2}\log^2\left(\frac{\lambda}{\Delta}\log\frac{1}{\eta}\right)\log\left(\frac{1}{\xi}\log\frac{\lambda}{\Delta}\right)\right)$ quantum circuits on $n + 1$ qubits. Each circuit uses one copy of $\rho$ and at most $\mathcal{O}\left(\frac{\lambda^2}{\Delta^2}\log^2 \frac{1}{\eta}\right)$ single-qubit Pauli rotations.
\end{theorem}



Thus, our quantum complexities are independent of the Hamiltonian sparsity $L$, at the price of the quadratic dependence $\widetilde{\mathcal{O}}(\lambda^2\Delta^{-2}\eta^{-2})$ for the total gate count. This is in contrast to standard results on phase estimation (see e.g., ~\cite[Table I]{Lee20}). Additionally, note that Theorem~\ref{thm:main} is derived using the specific choice of runtime vector in Eq.~\eqref{eq:constant-weight}. By tuning $\vec{r}$ (using for instance the optimisation procedures in Appendix~\ref{app:optimiser}), we can reduce the gate complexity per circuit by running more circuits, for a given set of problem parameters.



\section{Examples in quantum chemistry}\label{sec:examples}

\paragraph*{Comparisons.} 
Conventional phase estimation algorithms depend on the sparsity $L$, which is especially prohibitive for chemistry Hamiltonians. Several algorithms~\cite{berry2019qubitization,Lee20,vonburg2021DoubleFactorized} have used heuristic truncation policies to justify eliminating certain terms from the Hamiltonians, thereby reducing $L$.  While supporting numerics were presented, these truncations are not rigorous. Moreover, it was also assumed that only a single run of the algorithm suffices.  In practice, a single sample might return an incorrect result due to imperfect overlap with the ground state ($\eta < 1$), inherent failure probabilities of phase estimation, or quantum error correction failure events. In contrast, our algorithm is rigorously analysed; we use no Hamiltonian truncation, and upper-bound the number of samples needed in terms of $\eta$ and the target success probability.\footnote{We neglect quantum error correction failure events, though these can easily be suppressed to lower levels than other failure modes.}

\paragraph*{Hydrogen chains.}
As a benchmark system for assessing the scaling of quantum algorithms applied to quantum chemistry, we discuss hydrogen chains~\cite{Lee20, koridon2021orbital}. Using the best value $\lambda \sim \mathcal{O}(N^{1.34})$ given in \cite{koridon2021orbital}, our algorithm scales as $\widetilde{\mathcal{O}}(N^{2.68}/\Delta^2)$. For comparison, the scaling of qubitization is $\widetilde{\mathcal{O}}(N^{3.34}/\Delta)$ without truncation, and with heuristic truncations, $\widetilde{\mathcal{O}}(N^{2.3}/\Delta)$ for the sparse method of \cite{berry2019qubitization} and $\widetilde{\mathcal{O}}(N^{2.1}/\Delta)$ for the tensor hypercontraction approach of \cite{Lee20}. Hence, for constant $\Delta$, qubitization gives a better scaling than our algorithm if the proposed truncation schemes are accurate. However, we emphasize that our rigorous analysis does not make use of heuristic strategies for truncating Hamiltonian terms~\cite{Lee20, vonburg2021DoubleFactorized, berry2019qubitization} and that qubitization uses considerably more logical ancillae. Finally, if we are interested in extensive properties, where $\Delta \propto N$, then our approach scales as $\widetilde{\mathcal{O}}(N^{0.68})$, outperforming all qubitization algorithms.

\paragraph*{FeMoco.}
We estimate the costs of our algorithm applied to the Li \textit{et al.} FeMoco Hamiltonian~\cite{LiFeMoco}, another popular benchmark for which there have been several state-of-the-art resource studies~\cite{berry2019qubitization,Lee20,vonburg2021DoubleFactorized}.  We consider chemical accuracy $\Delta=0.0016$ Hartree, and use $\lambda = 1511$ Hartree, obtained using the bounds in~\cite{koridon2021orbital}.
We present our results in Fig.~\ref{fig:FeMoco}, illustrating the trade-off between the the expected number of gates per circuit and the number of samples required. Since the Hamiltonian from~\cite{LiFeMoco} has $152$ spin orbitals, each circuit uses $153$ qubits. 

\begin{figure}[t!]
    \centering
    \includegraphics[width=0.45
    \textwidth]{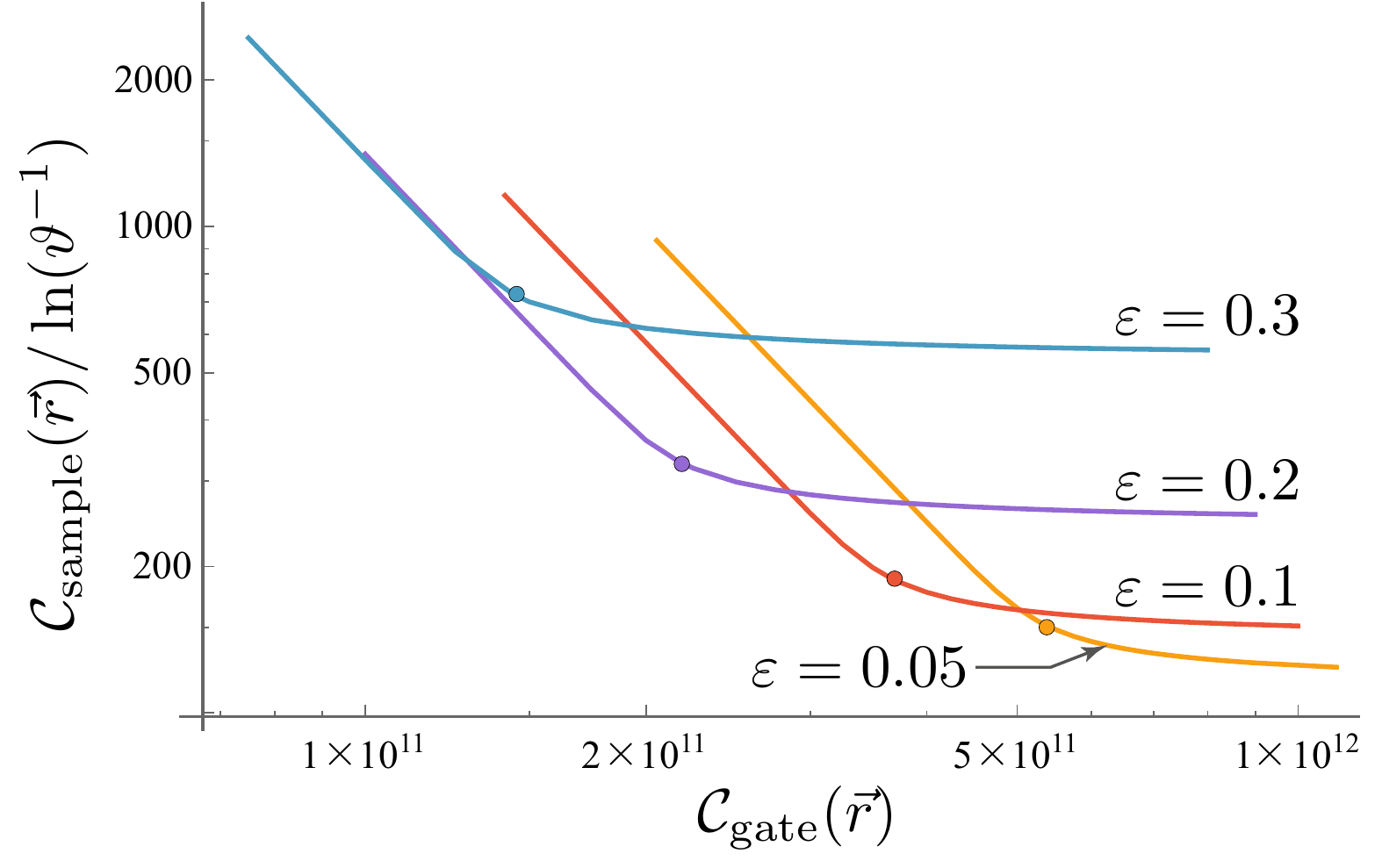}
    \caption{Log-log plot of $\mathcal{C}_{\mathrm{sample}}(\vec{r})/\ln(\vartheta^{-1})$ vs. $\mathcal{C}_{\mathrm{gate}}(\vec{r})$, for $\lambda = 1511$ (FeMoco~\cite{LiFeMoco,koridon2021orbital}), $\Delta= 0.0016$ (chemical accuracy), $\eta = 1$, and $\eps \in \{0.05, 0.1,0.2, 0.3\}$, for various choices of the runtime vector $\vec{r}$ calculated using the optimisation algorithms detailed in Appendix~\ref{app:optimiser}. Dots indicate the values that minimise the total expected complexity $2\mathcal{C}_{\mathrm{sample}}\cdot \mathcal{C}_{\mathrm{gate}}$, while curves are obtained by fixing $\mathcal{C}_{\mathrm{gate}}$ and minimising $\mathcal{C}_{\mathrm{sample}}$. To calculate the number of samples needed to guarantee an overall failure probability of $\leq \xi$ for ground state energy estimation (Theorem~\ref{thm:main}), one would multiply the $y$-axis by $\ln(\vartheta^{-1}) = \mathcal{O}\big(\log(\xi^{-1}) + \log\log(\delta^{-1})\big)$ (see \cite{lin21} for the explicit constants). As an example, $\xi = 0.1$ would give a multiplier of approximately $6$.}
    \label{fig:FeMoco}
\end{figure}

We have presented our gate counts as $\mathcal{C}_{\mathrm{gate}}$ controlled Pauli rotations, but {asymptotically our circuits can typically be realised using $\sim \! 2 \mathcal{C}_{\mathrm{gate}}$ Toffoli gates.  For modest system sizes and a modest number of logical ancilla ($\sim\! 40$), the Toffoli count is $\sim \! 6 \mathcal{C}_{\mathrm{gate}}$} (see Appendix~\ref{sec:GateComplexities}). The FeMoco resource estimate for the qDRIFT random compiler combined with phase estimation in~\cite[Appendix D]{Lee20} arrived at $10^{16}$ Toffoli gates per sample, which is $\sim\! 10^4$ times larger than $2 \mathcal{C}_{\mathrm{gate}}$ from the results in Fig.~\ref{fig:FeMoco}. Moreover, our rigorous analysis will likely be loose and overestimate resources; for instance, more aggressive---though heuristic---Hamiltonian rescaling is justifiable and can further reduce costs (see Appendix~\ref{App:TAU}).


\paragraph*{Acknowledgements.} We thank Sam McArdle for helpful discussions, especially with regard to calculations for the quantum chemistry examples, and Fernando Brand\~ao for discussions and support throughout this project.


\bibliography{random-phase}


\newpage
\onecolumngrid
\appendix



\section{Fourier series approximation to the Heaviside function (Lemma~\ref{lem:fourier})} \label{app:fourier}


In this appendix, we work toward a non-asymptotic version of Lemma \ref{lem:fourier}. Along the way, we provide various related approximation theory results with explicit constants. The main result will be Theorem \ref{thm:Fourier}, which shows that the Fourier series 
\begin{align} \label{Fbetad}
&F_{\beta,d}(x) = F_0 + \sum_{j=0}^{d} F_{2j+1} \left[e^{i(2j+1)x} - e^{-i(2j+1)x}\right]\\
&\text{with} \quad F_0 \coloneqq \frac{1}{2},\quad F_{2j+1} \coloneqq -i \sqrt{\frac{\beta}{2\pi}}e^{-\beta}\frac{I_j(\beta) + I_{j+1}(\beta)}{2j+1}\;\text{ for $0 \leq j \leq d-1$},\quad F_{2d+1} \coloneqq -i\sqrt{\frac{\beta}{2\pi}}e^{-\beta}\frac{I_d(\beta)}{2d+1} \label{eq:Fouriercoefficients}
\end{align}
can be made an arbitrarily good approximation to the Heaviside function $\Theta(x)$ on $x \in [-\pi,\pi]$ by choosing appropriate values for the parameters $\beta \in \mathbb{R}_{>0}$ and $d \in \mathbb{N}$. Here and throughout, $I_n(\cdot)$ denotes the $n^{\mathrm{th}}$ modified Bessel function of the first kind.


\subsection{Chebyshev approximation}

First, we construct an approximation $P_{\beta,d}(\cdot)$ to the Heaviside function in terms of Chebyshev polynomials; the properties of $P_{\beta,d}(\cdot)$ are characterised in Theorem~\ref{thm:Chebyshev} below. The following development is a strengthening of the results in \cite[Appendix A]{low17}, featuring more direct proofs as well as tighter constants. Note that the methods and bounds from \cite[Appendix A]{low17} were subsequently employed in e.g., \cite{gilyen19} and other works on 
quantum algorithms. 

We start with the following Chebyshev approximation to the scaled error function $\mathrm{erf}(\sqrt{2\beta}x):=\frac{2}{\sqrt{\pi}}\int_0^{\sqrt{2\beta}x}dt\,e^{-t^2}$, which in turn approximates the sign function for large $\beta$ (cf.~Lemma \ref{claim:sgnerf}). For $\beta \in \mathbb{R}_{>0}$ and $d \in \mathbb{N}$, we define
\begin{equation} \label{Qbetad}
Q_{\beta,d}(x) \coloneqq \sqrt{\frac{2\beta}{\pi}} 2e^{-\beta}\left[I_0(\beta)x + \sum_{j=1}^d I_j(\beta)(-1)^j \left(\frac{T_{2j+1}(x)}{2j+1} - \frac{T_{2j-1}(x)}{2j-1}\right)\right],
\end{equation}
where $T_n$ denotes the $n^{\mathrm{th}}$ Chebyshev polynomial of the first kind.

\begin{prop}  \label{claim:erfQ1}
For any $\beta \in \mathbb{R}_{> 0}$ and $d \in \mathbb{N}$, we have
\[ \max_{x\in [-1,1]}\left|\erf\big(\sqrt{2\beta}x\big) - Q_{\beta,d}(x) \right|\leq \frac{1}{d}\sqrt{\frac{2\beta}{\pi}}2e^{-\beta}\sum_{j=d+1}^\infty I_j(\beta). \]
\end{prop}

\begin{proof}
From the Chebyshev expansion of the error function (Proposition~\ref{claim:erfChebyshev}), we see that taking $d \to \infty$ in $Q_{\beta,d}(x)$ gives an exact expression for $\erf(\sqrt{2\beta}x)$. Hence,
\begin{align*}
\max_{x \in [-1,1]}\left|\erf\big(\sqrt{2\beta} x\big) - Q_{\beta,d}(x) \right| &= \max_{x \in [-1,1]}\left|2\sqrt{\frac{2\beta}{\pi}} e^{-\beta} \sum_{j=d+1}^\infty I_j(\beta) (-1)^j \left(\frac{T_{2j+1}(x)}{2j+1} - \frac{T_{2j-1}(x)}{2j-1} \right)\right| \\
&\leq 2\sqrt{\frac{2\beta}{\pi}}e^{-\beta} \sum_{j=d+1}^\infty I_j(\beta)\left(\frac{1}{2j+1} + \frac{1}{2j-1} \right) \\
&\leq 2\sqrt{\frac{2\beta}{\pi}}e^{-\beta} \frac{1}{d}\sum_{j=d+1}^\infty I_j(\beta)\,,
\end{align*}
using $\max_{x \in [-1,1]}|T_n(x)| = 1$ and the fact that 
\[ \frac{1}{2j+1} + \frac{1}{2j-1} \leq \frac{1}{j-1} \leq \frac{1}{d}\quad\forall\, j \geq d+1. \]
\end{proof}
Proposition~\ref{claim:erfQ1} shows that the error in using $Q_{\beta,d}(\cdot)$ to approximate the scaled error function depends on the infinite sum $\sum_{j=d+1}^\infty I_j(\beta)$ of modified Bessel functions. In the next proposition, we bound this sum directly in order to obtain tighter results than those given by \cite[Appendix A]{low17}, which used loose bounds from the survey \cite{sachdeva14}.

\begin{prop} \label{claim:besselsumbound}
For any $\beta \in \mathbb{R}_{> 0}$, $d \in \mathbb{N}$, and integer $t \geq \beta$, we have
\[ 2e^{-\beta}\sum_{j=d+1}^\infty I_j(\beta) \leq 2e^{-(d+1)^2/(2t)} + \frac{1}{2}\left(\frac{e\beta}{t}\right)^te^{-\beta}\,. \]
\end{prop}

\begin{proof}
Starting from the expression for $\sum_{j=d+1}^\infty I_j(\beta)$ given by Proposition~\ref{claim:besselsum}, we have
\begin{align*}
\sum_{j=d+1}^\infty I_j(\beta) &\leq \left[\sum_{j=d+1}^t + \sum_{j=t+1}^\infty\right]\frac{\beta^j}{j!} 2^{-j} \sum_{k=0}^{\floor{(j-d-1)/2}} {j \choose k} \\
&\leq \sum_{j=d+1}^t \frac{\beta^j}{j!}\exp\left[-\frac{(d+1)^2}{2j} \right] + \sum_{j=t+1}^\infty \frac{\beta!}{j!} 2^{-j} \cdot 2^{j-1} \\
&\leq \exp\left[-\frac{(d+1)^2}{2t}\right]\sum_{j=0}^\infty \frac{\beta^j}{j!} + \frac{1}{2}\sum_{j=t+1}^\infty \frac{\beta^j}{j!},
\end{align*}
where in the second line, we bound the first term using a Chernoff bound and the second term using the fact that the inner sum goes over fewer than half of the binomial coefficients. Hence, we find
\[ 2e^{-\beta} \sum_{j=d+1}^\infty I_j(\beta) \leq 2\exp\left[ -\frac{(d+1)^2}{2t}\right] + \sum_{j=t+1}^\infty \frac{\beta^j}{j!} e^{-\beta}\,. \]
The second sum on the right-hand side is an upper tail of the Poisson distribution with mean $\beta$, provided that $t + 1  > \beta$. In particular, it follows from \cite[Corollary 6]{makis13} that for $t \geq \beta$,
\begin{equation} \label{poissontail} \sum_{j=t+1}^\infty \frac{\beta^j}{j!}e^{-\beta} \leq \frac{1}{2}\left(\frac{e\beta}{t}\right)^t e^{-\beta}, \end{equation}
and the claim follows.
\end{proof}

We now define the function
\begin{equation} \label{eq:f}
f(\beta,\epsilon) \coloneqq \frac{\ln\left(\frac{1}{\varepsilon}\right) - \beta}{W\left(\frac{1}{e}\left[\frac{1}{\beta}\ln\left(\frac{1}{\varepsilon}\right) - 1 \right]\right)}
\end{equation}
for $\beta \in \mathbb{R}_{\geq 0}$ and $\varepsilon \in (0,1)$,
where $W(\cdot)$ denotes the principal branch of the Lambert-W function. As shown in Proposition~\ref{claim:Wthing},  $f(\beta, \epsilon)$ is the solution $t$ to the equation $(e\beta/t)^t e^{-\beta} = \varepsilon$ under the constraint $t > \beta$. 

\begin{lemma}[Chebyshev approximation to the error function] \label{claim:erfQ2}
For any $\beta,\varepsilon_1,\varepsilon_2 \in \mathbb{R}_{>0}$,
we have
\[ \max_{x \in [-1,1]}\left|\erf(\sqrt{2\beta}x) - Q_{\beta,d}(x) \right| \leq \varepsilon_1 + \varepsilon_2 \]
for any $d\in\mathbb{N}$ satisfying
\begin{equation} \label{d}
\text{$d \geq {\sqrt{t w_{\varepsilon_1}}}$ \quad where $w_{\varepsilon_1} \coloneqq W\left(\frac{8}{\pi \varepsilon_1^2} \right)$ and $t$ is any integer such that}
\end{equation}
\begin{equation} \label{t} t \geq \begin{dcases}
f\left(\beta, \sqrt{2\pi w_{\varepsilon_1}}\varepsilon_2\right)
&\qquad \text{if $\varepsilon_2 < \dfrac{1}{\sqrt{2\pi w_{\varepsilon_1}}}$}\\
\beta &\qquad \text{if $\varepsilon_2 \geq \dfrac{1}{\sqrt{2\pi w_{\varepsilon_1}}}$,}
\end{dcases} \end{equation}
with the function $f(\cdot,\cdot)$ defined as in \cref{eq:f}.
\end{lemma}

\begin{proof}
Combining Propositions~\ref{claim:erfQ1} and~\ref{claim:besselsumbound}, we have
\begin{equation} \label{erfQbound} \max_{x\in[-1,1]} \left|\erf\big(\sqrt{2\beta}x\big) - Q_{\beta,d}(x) \right| \leq \frac{1}{d}\sqrt{\frac{2\beta}{\pi}}\left[2e^{-(d+1)^2/(2t)} + \frac{1}{2}\left(\frac{e\beta}{t} \right)^t e^{-\beta} \right] \end{equation}
for any integer $t \geq \beta$. 
If we require that $t \geq \beta$, then to bound the first term on the right-hand side of Eq.~\eqref{erfQbound} by $\varepsilon_1$,
it suffices for
\[ \frac{1}{d} \sqrt{\frac{2t}{\pi}} 2e^{-d^2/(2t)} \leq \varepsilon_1\quad\Leftrightarrow\quad\frac{d^2}{t}e^{d^2/t} \geq \frac{8}{\pi \varepsilon_1^2}\,,\]
which holds for $d \geq \sqrt{tW(8/(\pi \varepsilon_1^2))}= \sqrt{t w_{\varepsilon_1}}$.
This in turn implies that $\sqrt{\beta}/d \leq \sqrt{t}/d \leq 1/\sqrt{w_{\varepsilon_1}}$, so the second term on the right-hand side of Eq.~\eqref{erfQbound} is at most $\varepsilon_2$ if
\[ \frac{1}{\sqrt{2\pi w_{\varepsilon_1}}} \left(\frac{e\beta}{t}\right)^t e^{-\beta} \leq \varepsilon_2. \]
The constraints on $t$ in Eq.~\eqref{t} then follow from Proposition~\ref{claim:Wthing}, together with the fact that Eq.~\eqref{erfQbound} holds only for $t \geq \beta$.
\end{proof}

For $\beta \in \mathbb{R}_{>0}$ and $d\in \mathbb{N}$, we define
\begin{equation} \label{Pbetad}
P_{\beta,d}(x) \coloneqq \frac{1}{2}\left(Q_{\beta,d}(x) + 1\right)\,.
\end{equation}
where $P_{\beta,d}(x)$ is a linear combination of Chebyshev polynomials, and approximates the Heaviside function $\Theta(x)$ with approximation error determined by $\beta$ and $d$ on a domain determined by $\beta$, as captured by the following theorem. 

\begin{theorem}[Chebyshev approximation to the Heaviside function] \label{thm:Chebyshev}
For any $\nu \in (0,1)$ and  $\varepsilon_1,\varepsilon_2,\varepsilon_3 \in \mathbb{R}_{>0}$, let $\beta$ be any real number such that $\beta \geq \frac{1}{4\nu^2}W\big(\frac{2}{\pi\varepsilon_3^2}\big)$ and let $d$ be any integer satisfying Eq.~\eqref{d} of Lemma \ref{claim:erfQ2}. Then,
\begin{equation} \label{chebthm1} \max_{x\in [-1,-\nu] \cup [\nu,1]} \left| \Theta(x)- P_{\beta,d}(x)\right| \leq \frac{1}{2}(\varepsilon_1 + \varepsilon_2 + \varepsilon_3), \end{equation}
and for all $x \in [-1,1]$,
\begin{equation} \label{chebthm2} -\frac{1}{2}(\varepsilon_1 + \varepsilon_2) \leq P_{\beta,d}(x) \leq 1 + \frac{1}{2}(\varepsilon_1 + \varepsilon_2). 
\end{equation}
\end{theorem}

\begin{proof}
By Eq.~\eqref{Pbetad}, we have
\begin{align*}
\max_{x\in [-1,-\nu] \cup [\nu,1]} \left| \Theta(x)- P_{\beta,d}(x)\right| &= \max_{x\in [-1,-\nu] \cup [\nu,1]} \left| \frac{1}{2}(\mathrm{sgn}(x) + 1)- \frac{1}{2}(Q_{\beta,d}(x) + 1) \right| \\
&\leq \frac{1}{2}\max_{x\in [-1,-\nu] \cup [\nu,1]} \Bigg( \left| \mathrm{sgn}(x) - \erf(\sqrt{2\beta}x) \right| + \left|\erf(\sqrt{2\beta}x) - Q_{\beta,d}(x) \right|\Bigg), 
\end{align*}
using the triangle inequality. Eq.~\eqref{chebthm1} now follows by applying Lemma~\ref{claim:erfQ2} to upper-bound the second term in the brackets by $\varepsilon_1 + \varepsilon_2$, and applying Lemma~\ref{claim:sgnerf} with $k = \sqrt{2\beta}$ to upper-bound the first term by $\varepsilon_3$. 

To prove Eq.~\eqref{chebthm2}, we note that
\begin{align*}
\max_{x\in [-1,1]} \left|Q_{\beta,d}(x) \right| \leq \max_{x \in [-1,1]}\left(\left|Q_{\beta,d}(x) - \erf(\sqrt{2\beta} x) \right| + \left|\erf(\sqrt{2\beta}x) \right| \right)\leq \varepsilon_1 + \varepsilon_2 + 1,
\end{align*}
using the triangle inequality and Lemma~\ref{claim:erfQ2}. 
Thus, for $x \in [-1,1]$,
\[ \frac{1}{2}[-(1+\varepsilon_1 + \varepsilon_2) + 1] \leq P_{\beta,d}(x) \leq \frac{1}{2}[(1+\varepsilon_1 + \varepsilon_2) + 1]. \]
\end{proof}


\subsection{Fourier approximation}

Next, we transform our polynomial approximation $P_{\beta,d}(\cdot)$ to the Heaviside function $\Theta(\cdot)$ into a Fourier series $F_{\beta,d}(\cdot)$, using the observation that $\Theta(x) = \Theta(\sin x)$ for all $x \in [-\pi,\pi]$. For $\beta \in \mathbb{R}_{>0}$ and $ d \in \mathbb{N}$, we define
\[ F_{\beta,d}(x) \coloneqq P_{\beta,d}(\sin x). \]
The following proposition allows us to extract the Fourier coefficients of $F_{\beta,d}(\cdot)$. 

\begin{lemma} \label{ChebtoFourier}
If $f(x) = \sum_{k=0}^\infty a_k T_k(x)$ for some $a_k \in \mathbb{C}$, then
\[ \text{$f(\sin x) = \sum_{k=0}^\infty F_k \left[e^{ikx} + (-1)^k e^{-ikx} \right]$ \quad with $F_k = \frac{1}{2}(-i)^k a_k$ $\forall\, k \in \mathbb{Z}_{\geq 0}$.} \]
\end{lemma}

\begin{proof}
Using the trigonometric identity $\sin x = \cos(\pi/2 - x)$ in conjunction with $T_k(\cos \theta) = \cos(k \theta)$, we have
\begin{align*} 
f(\sin x) 
= \sum_{k=0}^\infty a_k \cos\left(k \left(\frac{\pi}{2} - x \right) \right)
= \sum_{k=0}^\infty a_k \frac{1}{2}(-i)^k\left[e^{ikx} + (-1)^k e^{-ikx} \right]\,.
\end{align*}
\end{proof}
By reorganising the sum in Eq.~\eqref{Qbetad}, we have
\begin{align*}
P_{\beta,d}(x) 
&= \frac{1}{2} + \sqrt{\frac{2\beta}{\pi}} e^{-\beta} \left[\sum_{j=0}^{d-1}(I_j(\beta) + I_{j+1}(\beta))(-1)^j \frac{T_{2j+1}(x)}{2j+1} + I_d(\beta)(-1)^d \frac{T_{2d+1}(x)}{2d+1} \right].
\end{align*}
It then follows from Lemma~\ref{ChebtoFourier} that $F_{\beta,d}(x)\coloneqq P_{\beta,d}(\sin x)$ has the form given in Eq.~\eqref{Fbetad}, and we arrive at the following theorem.


\begin{theorem}[Fourier approximation to the Heaviside function] \label{thm:Fourier}
For any $\delta \in (0,\pi/2)$ and $\varepsilon_1,\varepsilon_2,\varepsilon_3 \in \mathbb{R}_{>0}$, let $\beta$ be any real number such that
\[ \beta \geq \max\left\{ \frac{1}{4\sin^2 \delta} W\left(\frac{2}{\pi\varepsilon_3^2}\right), 1\right\}\]
and let $d$ be any integer satisfying Eq.~\eqref{d}. Then, the function $F_{\beta,d}(\cdot)$ defined as in Eq.~\eqref{Fbetad} has the following properties:
\begin{enumerate}
    \item \hspace{1em}
    \begin{equation} \label{Ferror} \max_{x \in [-\pi + \delta, -\delta] \cup [\delta, \pi - \delta]}|\Theta(x) - F_{\beta,d}(x)| \leq \frac{1}{2}(\varepsilon_1 + \varepsilon_2 + \varepsilon_3), \end{equation}
    \item for all $x \in \mathbb{R}$,
    \begin{equation} \label{Fbound} -\frac{1}{2}(\varepsilon_1 + \varepsilon_2) \leq F_{\beta,d}(x) \leq 1 + \frac{1}{2}(\varepsilon_1 + \varepsilon_2), \end{equation}
    \item \hspace{1em}
    \begin{equation} \label{Fsum} \sum\limits_{j = 0}^d |F_{2j+1}| \leq \frac{1}{2}H_{d + {1}/{2}} + \ln 2, \end{equation}
\end{enumerate}
where $H_n$ denotes the $n^{\text{th}}$ harmonic number.
\end{theorem}

\begin{proof}
Eq.~\eqref{Ferror} follows from the Eq.~\eqref{chebthm1} of Theorem~\ref{thm:Chebyshev}. Specifically, since $\Theta(x) = \Theta(\sin x)$ for $x\in[-\pi,\pi]$ and $|\sin x| \geq \sin \delta$ for $x \in [-\pi + \delta, -\delta] \cup [\delta, \pi - \delta]$, we have
\begin{align*}
\max_{x \in [-\pi + \delta, -\delta] \cup [\delta, \pi - \delta]} \left|\Theta(x) - F_{\beta,d}(x) \right| &= \max_{x \in [-\pi + \delta, -\delta] \cup [\delta, \pi - \delta]}\left|\Theta(\sin x) - P_{\beta,d}(\sin x)\right| \\
&= \max_{y \in [-1,-\sin\delta] \cup [\sin\delta, 1]} \left|\Theta(y) - P_{\beta,d}(y) \right| \\
&\leq \frac{1}{2}(\varepsilon_1 + \varepsilon_2 + \varepsilon_3)
\end{align*}
by Theorem~\ref{thm:Chebyshev},
under the assumptions on $\beta$ and $d$. Eq.~\eqref{Fbound} follows immediately from Eq.~\eqref{chebthm2} of Theorem~\ref{thm:Chebyshev}, using the definition $F_{\beta,d}(x) \coloneqq P_{\beta,d}(\sin x)$ and the fact that $\sin x \in [-1,1]$ for all $x \in \mathbb{R}$. Finally, to prove Eq.~\eqref{Fsum}, we use Lemma~\ref{lem:kasperkovitz} from~\cite{kasperkovitz80} on the asymptotic behaviour of $I_n(\cdot)$ to bound
\begin{equation*}
\sqrt{2\pi \beta}e^{-\beta} I_j(\beta) \leq \exp(-j^2/(2\beta)) + 1.07 \beta^{-1/4} \leq 2.07\,.
\end{equation*}
By Eq.~\eqref{eq:Fouriercoefficients}, we have
\begin{align*}
    \sum_{j=0}^{d-1} |F_{2j+1}| = \sum_{j=0}^{d-1} \sqrt{\frac{\beta}{2\pi}} \frac{I_j(\beta) + I_{j+1}(\beta)}{2j+1} + \sqrt{\frac{\beta}{2\pi}}e^{-\beta} \frac{I_d(\beta)}{2d+1}\leq \frac{1}{2\pi}\sum_{j=0}^d  \frac{2\cdot 2.07}{2j+1} = \frac{2.07}{2\pi} \left(H_{d+1/2} + 2\ln 2 \right),
\end{align*}
and we loosen $2.07/(2\pi) < 1/2$.\footnote{Obtaining tight constants in the bound on the sum of Fourier coefficients is not crucial for our purposes, since when using our algorithm, one would numerically compute the Fourier coefficients and their sum. We remark that calculations in \cite[Appendix A.2]{gribling21} would lead to a result similar to ours, in Eq.~\eqref{Fsum}, for the coefficient sum.}
\end{proof}

Lemma~\ref{lem:fourier} from the main text is a special case of Theorem~\ref{thm:Fourier}, obtained by making the simple choice $\varepsilon_1 = \varepsilon_2 = \varepsilon_3 = 2\eps/3$. For completeness, we provide the proof below. Note that in practice, one can numerically minimise $d$ over the choice of $\varepsilon_1, \varepsilon_2, \varepsilon_3$, as we do in our numerical estimates in Sec.~\ref{sec:examples} of the main text.

\begin{proof}[Proof of Lemma~\ref{lem:fourier}] 
Choosing $\varepsilon_1 = \varepsilon_2 = \varepsilon_3 = 2\eps/3$ in Theorem~\ref{thm:Fourier}, we take
\[ \beta = \max\left\{\frac{1}{4\sin^2\delta} W\left(\frac{3}{\pi \eps^2}\right), 1 \right\} = \mathcal{O}\left(\frac{1}{\delta^2} \log\left(\frac{1}{\eps}\right)\right). \] Then, from Eqs.~\eqref{t} and~\eqref{fasymptotic}, we have that $t = \mathcal{O}(\beta + \ln(\eps^{-1})) = \mathcal{O}(\delta^{-2}\log(\eps^{-1}))$. We can choose $d = \ceil{\sqrt{tw_{\varepsilon_1}}}$, so since $w_{\varepsilon_1} = \mathcal{O}(\log(\eps^{-1}))$, we have $d = \mathcal{O}(\delta^{-1}\log(\eps^{-1})).$ 
Since $(\varepsilon_1 + \varepsilon_2 + \varepsilon_3)/2 = \eps$, it then follows from Eqs.~\eqref{Ferror} and~\eqref{Fbound} of  Theorem~\ref{thm:Fourier} that these choices of $\beta$ and $d$ ensure that $|\Theta(x) - F_{\beta,d}(x)| \leq \eps$ for all $x \in [-\pi + \delta, -\delta] \cup [\delta, \pi - \delta]$ and that $-\eps \leq F_{\beta,d}(x) \leq 1 +\eps$ for all $x \in \mathbb{R}$. Also, Eq.~\eqref{Fsum} of Theorem~\ref{thm:Fourier} gives $\sum_{j\in S_1} |F_j| = \mathcal{O}(\log d)$ for $S_1 = \{0\}\cup \{\pm (2j+1)\}_{j=0}^d$, since $F_0 = 1/2$ and the harmonic number scales as $H_n = \mathcal{O}(\log n)$.
\end{proof}



\subsection{Technical lemmas}\label{sec:fourier-misc}

The following results are used in the approximation theory proofs above.

\begin{prop}[Chebyshev expansion of the error function] \label{claim:erfChebyshev}
For any $\beta \in \mathbb{R}_{>0}$, we have
\[ \erf\big(\sqrt{2\beta} x\big) = 2\sqrt{\frac{2\beta}{\pi}}e^{-\beta} \left[I_0(\beta)x + \sum_{j=1}^\infty I_j(\beta)(-1)^j \left(\frac{T_{2j+1}(x)}{2j+1} - \frac{T_{2j-1}(x)}{2j-1}\right) \right]. \]
\end{prop}
\begin{proof}
By definition,
\begin{align*}
\erf\big(\sqrt{2\beta}x\big) = \frac{2}{\sqrt{\pi}}\int_0^{\sqrt{2\beta}x} dy\, e^{-y^2} = 2\sqrt{\frac{2\beta}{\pi}}\int_0^x dz\, e^{-2\beta z^2}= 2\sqrt{\frac{2\beta}{\pi}}e^{-\beta} \int_0^x dz\, e^{-\beta T_2(z)},
\end{align*}
changing variables in the second equality and substituting $T_2(z) = 2z^2 -1$ to obtain the third equality. Next, we use the fact that for any $x \in [-1,1]$, the Jacobi-Anger identity gives
\[ e^{-\beta x} = I_0(\beta) + 2\sum_{j=1}^\infty I_j(\beta) (-1)^jT_j(x). \]
Hence, since $T_2(z) \in [-1,1]$ for all $z \in [-1,1]$, we have that for any $x \in [-1,1]$,
\begin{align*}
\erf\left(\sqrt{2\beta}x\right) &= 2\sqrt{\frac{2\beta}{\pi}}e^{-\beta} \int_0^x dz\, \left(I_0(\beta) + 2\sum_{j=1}^\infty I_j(\beta) (-1)^n T_j(T_2(z)) \right) \\
&= 2\sqrt{\frac{2\beta}{\pi}}e^{-\beta} \int_0^x dz\, \left( I_0(\beta) + 2\sum_{j=1}^\infty I_j(\beta) (-1)^j T_{2j}(z)\right) \\
&= 2\sqrt{\frac{2\beta}{\pi}}e^{-\beta} \left[I_0(\beta)x + \sum_{j=1}^\infty I_j(\beta)(-1)^j \left(\frac{T_{2j+1}(x)}{2j+1} - \frac{T_{2j-1}(x)}{2j-1}\right) \right],
\end{align*}
where we used the composition identity $T_m(T_n(z)) = T_{mn}(z)$ in the second line and $\int dx\, T_n = \frac{1}{2}\left(\frac{T_{n+1}}{n+1} - \frac{T_{n-1}}{n-1} \right)$ in the third line, noting that $T_n(0) = 0$ for odd $n$. 
\end{proof}

\begin{prop} \label{claim:besselsum}
For any $\beta \in \mathbb{R}$ and $d \in \mathbb{Z}_{\geq 0}$, we have
\[ \sum_{j=d+1}^\infty I_j(\beta) = \sum_{j=d+1}^\infty \frac{\beta^j}{j!}2^{-j} \sum_{k=0}^{\floor{(j-d-1)/2}} {j\choose k}\,. \]
\end{prop}

\begin{proof}
We calculate
\begin{align*}
    \sum_{\ell=d+1}^\infty I_\ell(\beta) &= \sum_{\ell=d+1}^\infty\sum_{m=0}^\infty {2m+\ell \choose m} \frac{(\beta/2)^{2m+\ell}}{(2m+\ell)!} \\
    &= \sum_{\ell=d+1}^\infty \sum_{\substack{j = \ell\\\text{$j-\ell$ even}}}^\infty {j \choose (j - \ell)/2}\frac{(\beta/2)^j}{j!} \\
    &= \sum_{j=d+1}^\infty \sum_{\substack{\ell = d+1\\ \text{$j-\ell$ even}}}^j {j\choose (j-\ell)/2} \frac{(\beta/2)^j}{j!} \\
    &= \sum_{j=d+1}^\infty \frac{(\beta/2)^j}{j!} \sum_{k=0}^{\floor{(j-d-1)/2}} {j\choose k},
\end{align*}
where the first line follows from the definition of $I_\ell(\cdot)$, the second from making the change of variables $j = 2m+\ell$, the third from exchanging the summations, and the fourth from re-indexing the inner sum with $k = (j-\ell)/2$. 
\end{proof}

\begin{prop} \label{claim:Wthing}
For any $\beta, \varepsilon \in \mathbb{R}_{>0}$, we have
\begin{equation} \label{betaepsilon}
\left(\frac{e\beta}{t}\right)^t e^{-\beta} \leq \varepsilon.
\end{equation}
for all
\begin{equation} \label{t2} t \begin{dcases} \geq
f(\beta,\varepsilon)
\qquad &\text{if $\varepsilon < 1$} \\
\in \mathbb{R} \qquad &\text{if $\varepsilon \geq 1$},
\end{dcases} \end{equation}
where $f(\cdot,\cdot)$ is defined in Eq.~\eqref{eq:f}.
Moreover, $f(\beta, \varepsilon) > \beta$ for all $\varepsilon < 1$, and
\begin{equation} \label{fasymptotic}
f(\beta,\varepsilon) = \mathcal{O}\left(\beta + \ln\left(\frac{1}{\varepsilon}\right) \right)\,.
\end{equation}
\end{prop}

\begin{proof}
For any $\beta$, the function $g_\beta(t) \coloneqq (e\beta/t)^t e^{-\beta}$ reaches its global maximum of $1$ at $t = \beta$. Therefore, if $\varepsilon \geq 1$, Eq.~\eqref{betaepsilon} holds for all $t$. The function $g_\beta$ decreases monotonically past $t = \beta$, limiting to $0$ as $t \to \infty$. Hence, if $\varepsilon < 1$, we look for the $t > \beta$ such that $g_\beta(t) = \varepsilon$:
\begin{align*}
    \left(\frac{e\beta}{t}\right)^t e^{-\beta} = \varepsilon \quad \Leftrightarrow \quad \frac{1}{t}\left[\ln\left(\frac{1}{\varepsilon}\right) - \beta\right] \exp\left(\frac{1}{t}\left[\ln\left(\frac{1}{\varepsilon}\right) - \beta\right] \right) = \frac{1}{e\beta}\left[\ln\left(\frac{1}{\varepsilon}\right) - \beta\right].
\end{align*}
Note that since $\varepsilon <1$, the right-hand side of the last expression is always $> -1/e$. The solution such that $t > \beta$ is then given by the principal branch of the Lambert-W function:
\[ \frac{1}{t}\left[\ln\left(\frac{1}{\varepsilon}\right) - \beta\right] = W\left(\frac{1}{e}\left[\frac{1}{\beta}\ln\left(\frac{1}{\varepsilon}\right) - 1\right]\right),\]
which rearranges to give $t = f(\beta,\varepsilon)$. Thus, $f(\beta,\varepsilon) > \beta$. Eq.~\eqref{fasymptotic} follows from noting that since $\beta >0$, Eq.~\eqref{betaepsilon} holds for any $t$ such that $(e\beta/t)^t \leq \varepsilon$, and then applying a loosened version of \cite[Lemma 59]{gilyen19}.
\end{proof}

\begin{lemma}[Error function approximation to the sign function] \label{claim:sgnerf} 
For any $\nu, \varepsilon_3 \in \mathbb{R}_{> 0}$, and $k \in \mathbb{R}_{>0}$ satisfying $k^2 \geq \frac{1}{2\nu^2} W\left(\frac{2}{\pi\varepsilon_3^2}\right)$, we have
\[\max_{x \in (-\infty,-\nu] \cup [\nu,\infty)}\left|\mathrm{sgn}(x) - \erf(kx)\right| \leq \varepsilon_3\,. \]
\end{lemma}

\begin{proof}
\begin{align*} 
\max_{x \in (-\infty,-\nu] \cup [\nu,\infty)}\left|\mathrm{sgn}(x) - \erf(kx)\right| = \max_{x\in [\nu,\infty)} \left|1 - \erf(kx)\right| = \mathrm{erfc}(k\nu)\leq \frac{e^{-k^2\nu^2}}{\sqrt{\pi}k\nu}\,,
\end{align*}
where the first equality follows from the symmetry of $\sgn(\cdot)$ and $\erf(\cdot)$, the second from the fact that $1- \erf(kx) = \mathrm{erfc}(kx)$ is a decreasing function for $k>0$, and the third from the standard bound $\mathrm{erfc}(x) \leq e^{-x^2}/(\sqrt{\pi}x)$ for $x >0$. The upper bound is at most $\varepsilon_3$ if $2k^2\nu^2 e^{2k^2\nu^2} \geq {2}/({\pi \varepsilon_3^2})$, which holds for $2k^2\nu^2 \geq W(2/(\pi \varepsilon_3^2))$.
\end{proof}

\begin{lemma}\cite[Equation~7]{kasperkovitz80}\label{lem:kasperkovitz}
For any $\beta\geq1$ and $j\in\mathbb{Z}$, we have
\begin{align*}
\left|\sqrt{2\pi\beta}e^{-\beta}\cdot I_j(\beta)-\exp\left(-\frac{j^2}{2\beta}\right)\right|\leq1.07\beta^{-1/4}.
\end{align*}
\end{lemma}


\section{Proof of Eq.~\eqref{acdf}} \label{app:acdf}

In this appendix, we prove that our Fourier series approximation to the Heaviside function given in Eq.~\eqref{Fbetad} can be used to construct an approximate CDF, defined in Eq.~\eqref{Ctilde}, that satisfies the approximation guarantees in Eq.~\eqref{acdf}. This allows us to use the classical post-processing algorithms of Lin \& Tong \cite{lin21}. In particular, \cite[Appendix B]{lin21} proves a result similar to Proposition~\ref{prop:guarantee} below. However, their proof applies only under the assumption that the Fourier series $F(\cdot)$ is bounded as $0 \leq F(x) \leq 1$ for all $x \in \mathbb{R}$, whereas from Theorem~\ref{thm:Fourier}, we only have the guarantee $-\varepsilon \leq F(x) \leq 1 + \varepsilon$, for some $\varepsilon > 0$. Note that the Fourier series used in \cite{lin21} also only satisfies this weaker condition.

\begin{prop} \label{prop:guarantee}
Let $F(\cdot)$ be any function satisfying 
\begin{enumerate}
	\item $|\Theta(x) - F(x)| \leq \varepsilon$ for all $x \in [-\pi + \delta, -\delta] \cup [\delta, \pi-\delta]$, and
	\item $-\varepsilon \leq F(x) \leq 1 + \varepsilon$ for all $x \in [-\pi,\pi]$
\end{enumerate}
for some $\delta \in (0,\pi/2)$ and $\varepsilon \in \mathbb{R}_{>0}$. For any probability density function $p(\cdot)$ that is supported within the interval $[-(\pi-\delta)/2,(\pi-\delta)/2]$, define 
\[ C(x) \coloneqq (p *\Theta)(x), \qquad \widetilde{C}(x) \coloneqq (p * F)(x), \]
where $(f*g)(x) \coloneqq \int_{-\infty}^\infty dy\, f(y)g(x-y)$.
Then, for all $x \in [-(\pi-\delta)/2,(\pi-\delta)/2]$, 
\begin{equation} C(x-\delta) - \varepsilon \leq \widetilde{C}(x) \leq C(x + \delta) + \varepsilon. \end{equation}
\end{prop}
\begin{proof}
Let $x \in [-(\pi-\delta)/2,(\pi-\delta)/2]$. Since $p(y)$ is only non-zero for $y \in [-(\pi-\delta)/2,(\pi-\delta)/2]$ by assumption,
\begin{align*}
C(x-\delta) &= \int_{-(\pi-\delta)/2}^{(\pi-\delta)/2} dy\, p(y) (\Theta(x - \delta - y) - F(x-y))  
\\
&= \int_{x - (\pi-\delta)/2}^{x + (\pi-\delta)/2} dy'\, p(x - y')(\Theta(y' -\delta) - F(y')) 
\\
&= \left[\int_{x-(\pi-\delta)/2}^{-\delta} + \int_{-\delta}^\delta + \int_\delta^{x+(\pi-\delta)/2} \right] dy\, p(x-y)(\Theta(y-\delta) - F(y)),
\end{align*}
with the convention that $\int_a^b dy\, [\dots] = 0$ if $a > b$.
Since $x - (\pi-\delta)/2 \geq -(\pi-\delta)/2 - (\pi-\delta)/2 = -\pi-\delta$, we have $(x-(\pi-\delta)/2,-\delta) \subseteq (-\pi-\delta, -\delta)$. Similarly, since $x+(\pi-\delta)/2 \leq (\pi-\delta)/2 + (\pi-\delta)/2$, we have $(\delta, x+ (\pi-\delta)/2) \subseteq (\delta, \pi - \delta)$. Hence, $\Theta(y-\delta) - F(y) \leq \varepsilon$ in the first and third integrals, by assumption 1. For the second integral, note that $y-\delta \in (-2\delta,0)$ for $y \in (-\delta, \delta)$, so $\Theta(y-\delta) = 0$ on this range, and it follows from assumption 2 that $\Theta(y-\delta) - F(y) = -F(y) \leq \varepsilon$. Thus,
\begin{align*} 
C(x-\delta) - \widetilde{C}(x) &\leq \varepsilon \int_{x-(\pi-\delta)/2}^{x + (\pi-\delta)/2}dy\, p(x-y) = \varepsilon,
\end{align*}
so $C(x-\delta) - \varepsilon \leq \widetilde{C}(x)$. The upper bound is obtained using an analogous argument.
\end{proof}

When we defined the CDF $C(\cdot)$ in Eq.~\eqref{cdfdef}, we chose the normalisation factor $\tau = \frac{\pi}{2\lambda + \Delta}$, which implies that the corresponding probability density function $p(\cdot)$ is supported within the interval $[-\tau \lambda, \tau\lambda]$. Hence, we can apply Proposition~\ref{prop:guarantee} for any $\delta$ such that
\[ \tau\lambda \leq \frac{\pi-\delta}{2} \quad \Leftrightarrow \quad \delta \leq \pi - 2\lambda\tau = \tau\Delta. \]
This shows that the approximate CDF $\widetilde{C}(x) = (p*F)(x)$ from Eq.~\eqref{Ctilde} satisfies the guarantees in Eq.~\eqref{acdf} for all $\varepsilon \in \mathbb{R}_{>0}$ and $\delta \in (0,\lambda \tau]$, provided that we use the appropriate Fourier series $F(\cdot)$ from Lemma~\ref{lem:fourier} (with $\beta$ and $d$ in Eq.~\eqref{Fbetad} chosen appropriately in terms of $\delta$ and $\varepsilon$), as claimed in the main text.

\section{LCU decomposition of the time evolution operator (Lemma~\ref{lem:simulation})} \label{app:lcu}

In this appendix, we prove Lemma~\ref{lem:simulation} by constructing a particular decomposition of the time evolution operator $e^{i\hat{H}t}$ into a linear combination of unitaries (LCU). We then provide an algorithm, Algorithm~\ref{algsample}, for efficiently sampling a unitary from this decomposition with probability proportional to its coefficient.

\begin{proof}[Proof of Lemma~\ref{lem:simulation}]
By assumption, we have $\hat{H} = \sum_\ell p_\ell P_\ell$, where $p_\ell > 0$ for all $\ell$, $\sum_\ell p_\ell = 1$, and each $P_\ell$ is a Pauli operator. We write $e^{i\hat{H}t} = (e^{i\hat{H}t/r})^r$, and observe that if each $e^{i\hat{H}t/r}$ has an LCU decomposition
\begin{equation} \label{cmVm}
e^{i\hat{H}t/r} = \sum_m c_m W_m\,,
\end{equation}
then 
\begin{equation} e^{i\hat{H}t} = \sum_{m_1,\dots, m_r} c_{m_1}\dots c_{m_r} W_{m_1}\dots W_{m_r} \eqqcolon \sum_k b_k U_k \end{equation}
is an LCU decomposition for $e^{i\hat{H}t}$, with total weight
\begin{equation} \label{bkcm}
\sum_{k} |b_k| = \left(\sum_m |c_m| \right)^r\,.
\end{equation}
Furthermore, note that we can sample a unitary $U_k$ with probability proportional to $|b_k|$ by independently sampling $r$ unitaries $W_{m_1},\dots, W_{m_r}$ according to the distribution given by $\{|c_m|\}_m$, and implementing their product $W_{m_1}\dots W_{m_r}$.

We construct the following decomposition for $e^{i\hat{H}t/r}$. 
Letting $x \coloneqq t/r$, we have
\begin{align}
e^{i\hat{H}t/r} = e^{i\hat{H}x} &= \sum_{n=0}^\infty \frac{1}{n!}(ix\hat{H})^n \nonumber \\
&= \sum_{n \text{ even}} \frac{1}{n!}(ix\hat{H})^n\left(\mathbbm{1} + i\frac{x}{n+1} \hat{H}\right) \nonumber \\
&= \sum_{n \text{ even}} \frac{1}{n!} \left(ix \sum_\ell p_\ell P_\ell\right)^n \left(\mathbbm{1} + i\frac{x}{n+1}\sum_{\ell'} p_{\ell'} P_{\ell'}\right) \nonumber \\
&= \sum_{n \text{ even}} \frac{1}{n!} (ix)^n \sum_{\ell_1, \dots, \ell_n} p_{\ell_1} \dots p_{\ell_n} P_{\ell_1}\dots P_{\ell_n} \sum_{\ell'}p_{\ell'} \left(\mathbbm{1} + i\frac{x}{n+1} P_{\ell'} \right) \nonumber\\
&= \sum_{n \text{ even}} \frac{1}{n!} (ix)^n \sqrt{1 + \left(\frac{x}{n+1}\right)^2} \sum_{\ell_1, \dots, \ell_n, \ell'} p_{\ell_1}\dots p_{\ell_n} p_{\ell'} \left(P_{\ell_1}\dots P_{\ell_n} V_{\ell'}^{(n)} \right), \label{lcu1}
\end{align}
where 
\begin{align} \label{eq:angles}
V_{\ell'}^{(n)} \coloneqq \frac{1}{\sqrt{1 + \left(\frac{x}{n+1}\right)^2}} \left(\mathbbm{1} + i\frac{x}{n+1} P_{\ell'}\right)= \exp(i\theta_n P_{\ell'}), \quad\text{with }\theta_n \coloneqq \arccos\left(\left[1 + \left(\frac{x}{n+1}\right)^2 \right]^{-1/2}\right).
\end{align}
Thus, in the notation of Eq.~\eqref{cmVm}, we have decomposed $e^{i\hat{H}t/r}$ into a linear combination of unitaries of the form $W_m \to P_{\ell_1}\dots P_{\ell_n} V_{\ell'}^{(n)}$, with coefficients $c_m \to \frac{1}{n!} (ix)^{n}\sqrt{1 + \left(\frac{x}{n+1}\right)^2}p_{\ell_1} \dots p_{\ell_n} p_{\ell'}$. Since the $P_{\ell}$'s are Pauli operators, the controlled version of each of these unitaries $V_m$ can be implemented as a sequence of Cliffords (in particular, controlled Pauli operators) along with one multi-qubit Pauli rotation, which can be synthesised using Cliffords and one controlled single-qubit rotation. It follows that each unitary $U_k$ in the decomposition for $e^{i\hat{H}t}$ can be implemented using $r$ controlled single-qubit rotations. For the total weight of the coefficients, we have
\begin{align*} \sum_m |c_m|  = \sum_{n \text{ even}} \frac{1}{n!} |x|^n \sqrt{1 + \left(\frac{x}{n+1}\right)^2}\sum_{\ell_1, \dots_{\ell_n}, \ell'} p_{\ell_1}\dots p_{\ell_n} p_{\ell'}= \sum_{n=0}^\infty \frac{1}{(2n)!} x^{2n} \sqrt{1 + \left(\frac{x}{2n+1}\right)^2}
\end{align*}
with $x = t/r$. By Proposition~\ref{fact:sum} below, $\sum_m |c_m| \leq \exp(x^2) = \exp(t^2/r^2)$, so from Eq.~\eqref{bkcm}, $\sum_k |b_k| \leq \exp(t^2/r)$. Finally, we can make $b_k > 0$ for all $k$ by moving the phase $(i\sgn(x))^n = (i\sgn(t))^n$ of the coefficients $c_m$ onto the unitaries $W_m$.
\end{proof}

The non-Clifford gate complexity $r$ can in fact be chosen to be arbitrarily small, at the cost of an exponentially large sample complexity. Specifically, note that in the proof, $\sum_m |c_m|$ can also be bounded as $\sum_m |c_m| \leq \exp(x) = \exp(t/r)$, which gives $\sum_k |b_k| \leq \exp(t)$. An exponentially large sample complexity in the limit of zero non-Clifford gate complexity is consistent with the fact that Clifford circuits can be efficiently simulated classically~\cite{gottesman1998heisenberg,Aaronson2004}. 

The proof immediately gives the following algorithm for sampling from the decomposition from Lemma~\ref{lem:simulation}.

\begin{algorithm}[H]  \caption{Efficiently sample from the LCU decomposition of $e^{i\hat{H}t}$ from Lemma~\ref{lem:simulation}
} \label{algsample}
\textbf{Input:} 
A Hamiltonian $\hat{H} = \sum_{\ell=1}^L p_\ell P_\ell$ specified as a convex combination of Pauli operators, a real number $t$, a positive integer~$r$.
\\
\textbf{Output:} 
Description of a random unitary $U$, such that $\E[U] \propto e^{i\hat{H}t}$ and $U$ can be implemented using $r$ controlled single-qubit Pauli rotations. 
\begin{algorithmic}[1]
\State $\textsc{vList} \leftarrow ().$
\State \textbf{Do} $r$ times:
\State \hspace{1em} Sample an even non-negative integer $n$ with probability \[ q_n \propto \tfrac{(t/r)^n}{n!}\sqrt{1 + \left(\tfrac{t/r}{n+1}\right)^2}. \]
\State \hspace{1em} Independently sample $n+1$ indices $\ell_0,\dots, \ell_n$ from $\{p_\ell\}_{\ell=1}^L$.
\State \hspace{1em} Append $e^{i\theta_n P_{\ell_0}}$ to \textsc{vList}, where \[ \theta_n \coloneqq \arccos\Big(\Big[1 + \Big(\tfrac{t/r}{n+1}\Big)^2 \Big]^{-1/2}\Big). \]
\State \hspace{1em} Append $P_{\ell_1},\dots, P_{\ell_n}$ to \textsc{vList}.
\State \hspace{1em} Append $(i\sgn(t))^n \mathbbm{1}$ to \textsc{vList}.
\State \textbf{Return} $U = \textsc{vList}[l] \dots \textsc{vList}[2]\textsc{vList}[1]$, where $l = \mathrm{length}(\textsc{vList})$.
\end{algorithmic} 
\end{algorithm}

As presented, Step 3 of Algorithm~\ref{algsample} samples from an infinite distribution, over all even positive integers. 
However, the probability of sampling an integer $n$ decreases super-exponentially with $n$, so a very precise approximation can be made by truncating to only the first few terms in the distribution, as analysed in Appendix~\ref{App:Truncating}.

\begin{prop} \label{fact:sum}
For all $x \in \mathbb{R}$, we have
\[ \sum_{n=0}^{\infty}\frac{1}{(2n)!} x^{2n} \sqrt{1 + \left(\frac{x}{2n+1}\right)^2} \leq e^{x^2}. \]
\end{prop}

\begin{proof}
For $|x| \leq 1$, the sum of the first three terms in the series on the LHS can be bounded as
\begin{align*} 
\sum_{n=0}^2 \frac{1}{(2n)!} x^{2n} \sqrt{1 + \left(\frac{x}{2n+1}\right)^2} &\leq \sum_{n=0}^2 \frac{x^{2n}}{(2n)!}\left[1+ \frac{1}{2}\left(\frac{x}{2n+1}\right)^2 \right]\\
&= 1 + x^2 + \frac{5}{72}x^4 + \frac{1}{1200} x^6 \\
&\leq 1+ x^2 + \frac{1}{2}x^4 \\
&= \sum_{n=0}^2 \frac{1}{n!} x^{2n}\,,
\end{align*}
using $\sqrt{1 + a^2} \leq 1 + \frac{1}{2}a^2$, $x^6 \leq x^4$, and $5/72 + 1/1200 < 1/2$, while for all $n \geq 3$,
\[ \frac{1}{(2n)!} \sqrt{1 + \left(\frac{x}{2n+1}\right)^2} \leq \frac{1}{n!}. \]
Hence, we find
\begin{align*} \sum_{n=0}^\infty \frac{1}{(2n)!} x^{2n} \sqrt{1 + \left(\frac{x}{2n+1}\right)^2} \leq \sum_{n=0}^2 \frac{1}{n!} x^{2n} + \sum_{n=3}^\infty \frac{1}{n!} x^{2n}=  e^{x^2}
\end{align*}
for all $|x| \leq 1$. For $|x| > 1$, we have
\begin{align*}
\sum_{n=0}^{\infty}\frac{1}{(2n)!} x^{2n} \sqrt{1+ \left(\frac{x}{2n+1}\right)^2} &\leq \sum_{n=0}^\infty \frac{1}{(2n)!}x^{2n}\left(1 + \left|\frac{x}{2n+1}\right|\right)= \sum_{n=0}^\infty \frac{1}{n!} |x|^{n} = e^{|x|}< e^{x^2}\,,
\end{align*}
using $\sqrt{1+a^2} \leq 1+ |a|$ to obtain the first inequality.
\end{proof}

\section{Optimising the runtime vector for arbitrary Fourier series}
\label{app:optimiser}

In this appendix, we show how to optimise the complexity of our algorithm with respect to the runtime vector $\vec{r}$. 
Recall that in Algorithm \ref{alg1}, we use the LCU decomposition constructed in the proof of Lemma~\ref{lem:simulation} as our decomposition for each $e^{iHt_j}$, leading to Eq.~\eqref{eq:sampling}. Importantly, for each $j \in S_1$, we are free to choose any positive integer $r_j$ when applying Lemma~\ref{lem:simulation}. Then, the total weight of the coefficients of the decomposition is
\[ \mu_j \coloneqq \sum_{k \in S_2}b_k^{(j)} \leq \exp(t_j^2/r_j), \]
and the complexity of the controlled version of each unitary in the decomposition of $e^{iHt_j}$ is that of $r_j$ controlled single-qubit Pauli rotations. Specifically, from Eqs.~\eqref{eq:A-sum}, \eqref{Csample}, and~\eqref{Cgate}, we have the expressions
\begin{align}
    &\mathcal{A}(\vec{r}) = \sum_{j \in S_1} |F_j| \mu_j, \label{A_app}\\
    &\mathcal{C}_{\mathrm{sample}}(\vec{r}) = \left\lceil \left(\frac{2\mathcal{A}(\vec{r})}{\eta/2-\varepsilon}\right)^2\ln\frac{1}{\vartheta}\right\rceil \label{Csample_app} \\
    &\mathcal{C}_{\mathrm{gate}}(\vec{r}) = \frac{1}{\mathcal{A}(\vec{r})}\sum_{j \in S_1} |F_j| \mu_j r_j \label{Cgate_app}.
\end{align}
Note that for each $j$, the value of $\mu_j$ implicitly depends on $r_j$. Also, observe that we can in fact exclude all indices $j$ such that $t_j = 0$ from the sums in Eqs.~\eqref{A_app} and Eq.~\eqref{Cgate_app}. This is due to the  fact that $\tr[\rho e^{0}] =1$ does not have to be estimated, so we do not have to sample these indices when implementing Algorithm~\ref{alg1}. Therefore, throughout this section, all sums will implicitly be over $S_1 \setminus \{j: t_j = 0\}$.

When finding the optimal runtime vector $\vec{r} = (r_j)_{j\in S_1} \in \mathbb{N}^{S_1}$, which determines the sample and gate complexities, we consider two different goals:
\begin{enumerate}
\item finding $\vec{r}$ that minimises the expected total gate complexity $\mathcal{C}_{\mathrm{total}}(\vec{r})=2\mathcal{C}_{\mathrm{sample}}(\vec{r}) \cdot \mathcal{C}_{\mathrm{gate}}(\vec{r})$ (Appendix \ref{app:total}),
\item minimising the sample complexity given an upper bound on the expected gate complexity per sample, i.e., given a number $g \in \mathbb{R}_{>0}$, finding $\vec{r}$ that minimises $\mathcal{C}_{\text{sample}}(\vec{r})$ subject to the constraint $\mathcal{C}_{\mathrm{gate}}(\vec{r}) \leq g$ (Appendix \ref{app:sample}).
\end{enumerate}
Our Algorithm~\ref{alg1} uses the specific Fourier series of Lemma~\ref{lem:fourier}, whose coefficients are specified in Appendix~\ref{app:fourier}, but we note that the results presented in this section will apply to any arbitrary set $\{F_j\}_{j\in S_1}$ of Fourier coefficients. 

We solve both problems approximately, where the sources of approximation are as follows.
\begin{itemize}
\item We use approximate expressions for the complexities, replacing $\mu_j$ with its analytic upper bound
\begin{equation} \label{eq:MukBound} u_j \coloneqq \exp(t_j^2/r_j) \end{equation}
in the formulae for  $\mathcal C_{\mathrm{sample}}(\vec{r})$, $\mathcal C_{\mathrm{gate}}(\vec{r})$, and $\mathcal C_{\mathrm{total}}(\vec{r})$.
\item  We ignore the fact that $\mathcal C_{\mathrm{sample}}(\vec{r})$ must be an integer, i.e., we remove the ceiling in Eq.~\eqref{Csample_app}.
\item The $r_j$'s that minimise the approximation expressions are real numbers in general. On the other hand, in the context of our algorithm, the $r_j$'s are required to be integers. We simply round our results for the approximate $r_j$'s to the nearest integer.
\end{itemize}
After determining the $\vec{r}$ that optimise these approximate expressions, we then round each entry and substitute the resulting vector into the \emph{exact} expression, Eqs.~\eqref{Csample_app} and~\eqref{Cgate_app}, in our numerical calculations. This gives complexities that are valid and exact, but only near-optimal in general. However, the upper bound $u_j$ for $\mu_j$ becomes tight for small $t_j/r_j$, and the $r_j$'s are typically large enough that the rounding error is negligible, so we expect these complexities to be close to the optimal solutions.

Note that since $\mathcal{C}_{\mathrm{gate}}(\vec{r})$ does not include the cost of preparing the ansatz state $\rho$, in this section we are considering only the case where the state preparation complexity is small relative to the cost of implementing the random unitaries from Lemma~\ref{lem:simulation}. However, the same techniques can be straightforwardly adapted to incorporate state preparation costs into the optimisation.


\subsection{Minimising the total complexity}\label{app:total}


Making the approximations discussed above, we have
\begin{align*} 
\mathcal{C}_{\mathrm{total}}(\vec{r}) \propto \mathcal{A}(\vec{r}) \sum_{k}  |F_k| \mu_k r_k \leq \Bigg( \sum_{j} |F_j| u_j \Bigg)\left(\sum_k |F_k| u_k r_k \right) \eqqcolon c(\vec{r}),
\end{align*}
where we have used that $\mu_j \leq u_j$ for all $j$ (cf.~Eq.~\eqref{eq:MukBound}).

The following proposition reduces the minimisation of $c(\vec{r})$, and hence $\mathcal{C}_{\mathrm{total}}(\vec{r})$, to a simple one-dimensional optimisation problem.


\begin{prop} \label{claim:optimalrs}
For any $F_j \in \mathbb{C} \setminus \{0\}$ and $t_j \in \mathbb{R} \setminus \{0\}$ for $j \in S_1$, define the function $c: \mathbb{R}_{>0}^{S_1} \to \mathbb{R}_{>0}$ by
\[ c(\vec{r}) = \Bigg(\sum_j |F_j| u_j\Bigg)\left(\sum_k |F_k|u_k r_k \right),\]
where for each $j \in S_1$,
\begin{equation} \label{u_j} u_j = u_j(r_j) \coloneqq \exp(t_j^2/r_j). \end{equation}
Then, $\vec{r}\,^* \coloneqq \argmin\limits_{\vec{r} \in \mathbb{R}_{>0}^{S_1}} c(\vec{r})$ satisfies
\begin{equation} \label{r_jstar}
r_j^* = \frac{t_j^2}{2}\left(1 + \sqrt{1+ \frac{4}{t_j^2}S(\vec{r}\,^*)} \right) \quad \forall\, j \in S_1 \setminus \{0\},
\end{equation}
where $S: \mathbb{R}_{>0}^{S_1} \to \mathbb{R}_{>0}$ is the function
\begin{equation} \label{S} S(\vec{r}) \coloneqq \frac{\sum_j |F_j| u_j r_j}{\sum_k |F_k| u_k}. \end{equation}
\end{prop}
\begin{proof}
Using $\partial u_j/\partial r_l = -\delta_{jl} t_l^2 u_l/r_l^2$,
\begin{align*}
\frac{\partial c}{\partial r_l} = |F_l|\left(-\frac{t_l^2}{r_l^2} u_l\right)\left( \sum_k |F_k| \mu_k r_k\right) + \Bigg(\sum_j |F_j| u_j\Bigg)|F_l|\left(1- \frac{t_l^2}{r_l}\right)u_l 
\end{align*}
which is $0$ if and only if
\[ -\frac{t_l^2}{r_l^2} S(\vec{r}) + 1 - \frac{t_l^2}{r_l} = 0. \]
Solving this for $r_l$ and taking the positive root leads to Eq.~\eqref{r_jstar}. It is easily verified that $c(\cdot)$ attains its global minimum here. 
\end{proof} 

Note that Proposition~\ref{claim:optimalrs} reduces the \textit{a priori} multi-dimensional problem of minimising the multivariate function ${c}(\cdot)$ to a simple one-dimensional problem. To see this, define the function $\vec{R}: \mathbb{R}_{>0} \to \mathbb{R}_{>0}^{S_1}$ by $R_j(s) \coloneqq (t_j^2/2)(1 + \sqrt{1 + 4s/t_j^2})$ for all $j$, so that $\vec{r}\,^* = \vec{R}(S(\vec{r}\,^*))$ by Proposition~\ref{claim:optimalrs}. Applying the function $S(\cdot)$ to both sides, we see that $s^* \coloneqq S(\vec{r}\,^*)$ satisfies the single-variable equation
\begin{equation} \label{sstar}
s^* = S(\vec{R}(s^*)),
\end{equation}
which can be solved numerically using standard root-finding methods. In particular, by using the fact that the function $(x/2)(1+ \sqrt{1 + 4s/x})$ increases with $x$ for any $s \in \mathbb{R}$, it can be shown that $0 < s^* \leq 2t_{\max}^2$, where $t_{\max} \coloneqq \max_{j\in S_1} |t_j|$. Hence, using the bisection method, for instance, $s^*$ can be found to within additive error $\epsilon$ in $\mathcal{O}(\log(t_{\max}/\epsilon))$ iterations.
Upon finding $s^*$, one can then easily calculate every $r_j^*$ as $r_j^* = R_j(s^*)$.


\subsection{Minimising the sample complexity given constraints on gate complexity}\label{app:sample}

Above, we considered minimising $\mathcal{C}_{\mathrm{total}}(\vec{r}) = \mathcal{C}_{\mathrm{sample}}(\vec{r}) \cdot \mathcal{C}_{\mathrm{gate}}(\vec{r})$, with no constraints on $\mathcal{C}_{\mathrm{sample}}$ or $\mathcal{C}_{\mathrm{gate}}$. In some scenarios, it may be useful to consider the following more constrained problem: Impose an upper bound on $\mathcal{C}_{\mathrm{gate}}(\vec{r})$, the expected gate complexity per Hadamard test. 
What choice of runtime vector $\vec{r}$ gives the minimum sample complexity $\mathcal{C}_{\mathrm{sample}}(\vec{r})$? This problem may be well-motivated in e.g., the context of ``early fault-tolerance,'' where it might be advantageous to run a larger number of shorter circuits, even if this increases the total complexity~\cite{lin21}.

We approximate
\[ \mathcal{C}_{\mathrm{sample}}(\vec{r}) \approx \frac{4\ln(1/\vartheta)}{(\eta/2 - \varepsilon)^2}\Bigg(\sum_j |F_j| u_j \Bigg)^2, \qquad \mathcal{C}_{\mathrm{gate}}(\vec{r}) \approx \frac{\sum_j |F_j| u_j r_j}{\sum_k |F_k| u_k}= S(\vec{r})\,,\]
where $u_j$ is defined in Eq.~\eqref{u_j} and $S(\cdot)$ in Eq.~\eqref{S}.
It is clear that allowing for larger $\mathcal{C}_{\mathrm{gate}}$ decreases the minimum $\mathcal{C}_{\mathrm{sample}}$ required, so we replace our upper limit on $\mathcal{C}_{\text{gate}}$ (an inequality constraint) by an equality constraint. 

\begin{prop} \label{claim:rsG}
For any $F_j \in \mathbb{C} \setminus \{0\}$ and $t_j \in \mathbb{R}\setminus \{0\}$ for $j \in S_1 \setminus \{0\}$, let $u_j$ and $S(\cdot)$ be defined as in Eqs.~\eqref{u_j} and~\eqref{S}. For any $g > 0$, define the function $\vec{R}^{(g)}$ by
\[ R^{(g)}_j(\lambda) \coloneqq \frac{t_j^2}{2}\left[1 + \sqrt{1 + \frac{4}{t_j^2}\left(\frac{1}{\lambda} - g \right)}\right].  \]
Then, the minimum of $f(\vec{r}) \coloneqq \sum_j |F_j| u_j$ subject to $S(\vec{r}) = g$ and $\vec{r} > 0$ is attained by 
\begin{equation} \label{rstarg} \vec{r}\,^{(g)*} \coloneqq \vec{R}^{(g)}(\lambda^*), \end{equation}
where $\lambda^*$ is the solution to
\begin{equation} \label{lambdastar}
S(\vec{R}^{(g)}(\lambda^*)) = g\,.
\end{equation}
\end{prop}

\begin{proof}
The constraint $S(\vec{r}) = g$ is equivalent to 
\[ \sum_j|F_j| u_j r_j = g\sum_j |F_j| u_j, \]
so using the method of Lagrange multipliers, we solve
\[ \vec{\nabla}\mathcal{L}(\vec{r}\,^{(g)*}, \lambda^*) = \vec{0} \]
for $\vec{r}\,^{(g)*}$ and $\lambda^*$, where 
\[ \mathcal{L}(\vec{r}, \lambda) \coloneqq \sum_j |F_j| u_j - \lambda\Bigg(g\sum_j |F_j| u_j - \sum_j |F_j| u_j r_j\Bigg). \]
This gives
\[ {r}_j^{(g)*} = \frac{t_j^2}{2}\left[1 + \sqrt{1 + \frac{4}{t_j^2}\left(\frac{1}{\lambda^*} - g\right)} \right]\]
for all $j$ (under the requirement that $\vec{r}\,^{(g)*} > 0$) and 
\[ S(\vec{r}\,^{(g)*}) = g, \]
which when rewritten in terms of the function $\vec{R}^{(g)}$ gives the result. 
\end{proof}
Thus, we can find all of the $r_j^{(g)*}$'s by solving a one-dimensional problem---namely, by solving Eq.~\eqref{lambdastar} for $\lambda^*$---then simply evaluating $R_j^{(g)}(\lambda^*)$ for each $j$.
An analogous strategy can be used to find the $\vec{r}$ that minimises the expected gate complexity given a fixed upper bound on the sample complexity.


\section{Probabilistic analysis} \label{app:prob}

In this section, we fill in the details for the proof that Algorithm~\ref{alg1} outputs an incorrect answer with probability at most $\vartheta$, and analyse the effect of truncating the infinite distribution in Algorithm~\ref{algsample} to a small number of terms.

\subsection{Failure probability of Algorithm~\ref{alg1}}


Let all quantities be defined as in Algorithm~\ref{alg1}. 
Recall that by Eq.~\eqref{acdf}, $\widetilde{C}(x) < \eta - \eps$ would imply that $C(x-\delta) < \eta$, while $\widetilde{C}(x) > \eps$ would imply that $C(x + \delta) > 0$, so we can solve Problem~1 by deciding between $\widetilde{C}(x) < \eta - \eps$ and $\widetilde{C}(x) > \eps$, outputting either if both are true. To estimate $\widetilde{C}(x)$, Steps 5-7 of Algorithm~\ref{alg1} sample from the decomposition
\[ \widetilde{C}(x) = \sum_{(j,k) \in S_1 \times S_2} a_{jk}\tr[\rho U_k^{(j)}] \]
from Eq.~\eqref{eq:sampling} as follows. Let $J,K$ denote random variables with $\mathbb{P}[J = j, K = k] = |a_{jk}|/\mathcal{A}(\vec{r})$, where $\mathcal{A}(\vec{r}) = \sum_{(j,k) \in S_1 \times S_2}|a_{jk}|$ as in Eq.~\eqref{eq:A-sum} (the $a_{jk}$ depend implicitly on the runtime vector $\vec{r}$). For each unitary $U_{k}^{(j)}$, let $X_{jk}$ and $Y_{jk}$ denote the random variables associated with the outcomes of the Hadamard test on $\rho$ and $U_k^{(j)}$, such that $\E[X_j] = \mathrm{Re}(\tr[\rho U_k^{(j)}])$ and $\E[Y_j] = \mathrm{Im}(\tr[\rho U_k^{(j)}])$. Then, the random variable
\[ Z \coloneqq \mathcal{A}(\vec{r})e^{i\,\mathrm{arg}(a_{JK})}(X_{JK} + iY_{JK}) \]
is a unbiased estimator for $\widetilde{C}(x)$, i.e., $\E[Z] = \widetilde{C}(x)$. To compare to Step~7 of Algorithm~\ref{alg1}, note from Eq.~\eqref{eq:sampling} that $\mathrm{arg}(a_{jk}) = \mathrm{arg}(F_j) + jx$, since $b_k^{(j)} > 0$ by Lemma~\ref{lem:simulation}. 

Step~8 of Algorithm~\ref{alg1} computes the average $\overline{Z}$ of $\mathcal{C}_{\mathrm{sample}}(\vec{r})$ independent samples of $Z$, and compares its real part to $\eta/2$, guessing that $\widetilde{C}(x) < \eta - \eps$ (so $C(x-\delta) < \eta$) if $\overline{Z} < \eta/2$ and that $\widetilde{C}(x) > \eps$ (so $C(x+\delta) > 0$) if $\overline{Z} \geq \eta/2$. (Here, we write e.g., $\overline{Z} < a$ as shorthand for $\mathrm{Re}(\overline{Z}) <a$.) Thus, the probability of error is bounded as
\begin{align*} p_{\mathrm{error}} &\leq \Pr\left[\overline{Z} < \eta/2 \, \Big| \, \widetilde{C}(x) \geq \eta - \eps\right]\Pr\left[\widetilde{C}(x) \geq \eta-\eps\right] + \Pr\left[\overline{Z} \geq \eta/2 \, \Big| \, \widetilde{C}(x) \leq \eps\right]\Pr\left[\widetilde{C}(x) \leq \eps\right] \\
& \leq \Pr\left[\overline{Z} < \eta/2 \, \Big| \, \widetilde{C}(x) = \eta - \eps\right]\Pr\left[\widetilde{C}(x) \geq \eta-\eps\right] + \Pr\left[\overline{Z} \geq \eta/2 \, \Big| \, \widetilde{C}(x) = \eps\right]\Pr\left[\widetilde{C}(x) \leq \eps\right].
\end{align*}
Since $\E[\overline{Z}] = \widetilde{C}(x)$, we see that conditioned on $\widetilde{C}(x) = \eta-\varepsilon$, 
\begin{equation} \label{tail} \Pr\left[\overline{Z} < \eta/2\right] = \Pr\bigg[ \E[\overline{Z}] - \overline{Z} > \eta/2 - \eps\bigg]. \end{equation}
Hence, using the fact that $|Z| = \sqrt{2}\mathcal{A}(\vec{r})$, so $\mathrm{Re}(Z)$ is contained in the interval $[-\sqrt{2}\mathcal{A}(\vec{r}),\sqrt{2}\mathcal{A}(\vec{r})]$, Hoeffding's inequality gives
\begin{align*}
    \Pr\left[\overline{Z} < \eta/2 \, \Big| \, \widetilde{C}(x) = \eta - \eps\right] &\leq \exp\left[-\frac{2(\eta/2-\eps)^2}{(2\sqrt{2}\mathcal{A}(\vec{r}))^2} \mathcal{C}_{\mathrm{sample}}(\vec{r}) \right],
\end{align*}
and the right-hand side is at most $\vartheta$ when $\mathcal{C}_{\mathrm{sample}}(\vec{r})$ is chosen as in Eq.~\eqref{Csample} of Algorithm~\ref{alg1}. Likewise, $\Pr\left[\overline{Z} \geq \eta/2 \, \Big| \, \widetilde{C}(x) = \eps\right] \leq \eps$, so
\[ p_{\mathrm{error}} \leq \vartheta \left(\Pr\left[\widetilde{C}(x) > \eta-\eps\right] + \Pr\left[\widetilde{C}(x) \leq \eps\right] \right) \leq \vartheta, \]
as claimed.

\subsection{Truncating the infinite distribution}
\label{App:Truncating}

The linear decomposition of $e^{i\hat{H}t}$ into unitaries constructed in the proof of Lemma~\ref{lem:simulation} in Appendix~\ref{app:lcu} contains an infinitely many unitaries, due to the fact that the index $n$ in the sum in Eq.~\eqref{lcu1} ranges over all even, non-negative integers. In practice, instead of sampling from a distribution over infinitely many integers in Step~3 of Algorithm~\ref{algsample}, one could truncate the sum in Eq.~\eqref{lcu1} at some order $M$, and modify Step~3 accordingly. Since the coefficients in Eq.~\eqref{lcu1} decay very rapidly with $n$, the effect of this truncation is negligible for modest values of $M$.

To make this rigorous, we consider the random variable $Z'$ that results from sampling from the truncated distribution, and see how this changes the Hoeffding's inequality analysis in the previous subsection. The total weight of the truncated LCU will be less than $\mathcal{A}(\vec{r})$, so $\mathrm{Re}(Z')$ is bounded in the interval $[-\sqrt{2}\mathcal{A}(\vec{r}),\sqrt{2}\mathcal{A}(\vec{r})]$. Unlike $Z$, however, $Z'$ is not in general an unbiased estimator for $\widetilde{C}(x)$; we will have $\E[Z] - \E[Z'] = \widetilde{C}(x) - \E[Z'] = B$ for some bias $B$. Consequently, defining $\overline{Z'}$ to be the average of $\mathcal{C}_{\mathrm{sample}}(\vec{r})$ independent samples of $Z'$, the analogue of Eq.~\eqref{tail} would read
\[ \Pr[\overline{Z'} < \eta/2] \leq \Pr \bigg[\E[\overline{Z'}] - \overline{Z'} > \eta/2 - \varepsilon - |B| \bigg] \]
(using $\overline{Z'}$ as shorthand for $\mathrm{Re}(\overline{Z'})$ where it is clear from context), leading to
\begin{align*}
    \Pr\left[\overline{Z'} < \eta/2 \, \Big| \, \widetilde{C}(x) = \eta - \eps\right] &\leq \exp\left[-\frac{2(\eta/2-\eps - |B|)^2}{(2\sqrt{2}\mathcal{A}(\vec{r}))^2} \mathcal{C}_{\mathrm{sample}}(\vec{r}) \right],
\end{align*}
by Hoeffding's inequality, and similarly for $\Pr\left[\overline{Z'} \geq \eta/2 \, \Big| \, \widetilde{C}(x) = \eps\right] \leq \eps$. Therefore, it suffices to replace the $(\eta/2 - \eps)^{-2}$ factor in the definition of $\mathcal{C}_{\mathrm{sample}}(\vec{r})$ in Eq.~\eqref{Csample} with $(\eta/2 - \eps - |B|)^{-2}$.

Hence, it remains to bound the bias $B$. We show in Theorem~\ref{thm:truncate} below that $|B|$ decreases superexponentially with the truncation order $M$. For this, we introduce some extra notation for convenience. From the proof of Lemma~\ref{lem:simulation}, we have that for each $j \in S_1$,
\begin{equation*}
e^{i\hat{H}t_j/r_j} = \sum_{\substack{n=0\\n \text{ even}}}^\infty \alpha_n^{(j)} A_n^{(j)}
\end{equation*}
with
\begin{align*}
    \alpha_n^{(j)} \coloneqq \frac{1}{n!} \left(\frac{t_j}{r_j}\right)^n\sqrt{1 + \left(\frac{t_j/r_j}{n+1}\right)^2}, \qquad A_n^{(j)} = (i\hat{H})^n \frac{1}{\sqrt{1 + \left(\frac{t_j/r_j}{n+1}\right)^2}}\left(\mathbbm{1} + i\frac{t_j/r_j}{n+1}\hat{H}\right).
\end{align*}
Since $\hat{H} = \sum_\ell p_\ell P_\ell$ is assumed to be a convex combination of Pauli operators, each $A_n^{(j)}$ is a convex combination of unitaries. Note also that 
in this notation, the total weight $\mu_j$ of the coefficients in the LCU decomposition for $e^{i\hat{H}t_j} = (e^{i\hat{H}t_j})^{r_j}$ is given by
\begin{equation} \label{mu_j2} \mu_j = \Bigg(\sum_{\substack{n = 0\\n\text{ even}}}^\infty \alpha_n^{(j)} \Bigg)^{r_j}\end{equation}
for each $j$. Then, we have 
\[ \E[Z] = \sum_{j \in S_1} F_j \tr\Bigg[ \rho \Bigg(\sum_{\substack{n=0\\\text{$n$ even}}}^\infty \alpha_n^{(j)}A_n^{(j)}\Bigg)^{r_j}\Bigg] = \widetilde{C}(x), \qquad \E[Z'] = \sum_{j \in S_1} F_j \tr\Bigg[ \rho \Bigg(\sum_{\substack{n=0\\\text{$n$ even}}}^M \alpha_n^{(j)}A_n^{(j)}\Bigg)^{r_j}\Bigg]. \]

\begin{theorem} \label{thm:truncate} For any $F_j \in \mathbb{C}$, $t_j \in \mathbb{R}$, and $\vec{r} \in \mathbb{N}^{S_1}$ such that $r_j \geq |t_j|$ for all $j$,\footnote{Note that for $|t_j| \geq 1$, this is satisfied by both the heuristic choice in Eq.~\eqref{eq:constant-weight} as well as the near-optimal solutions given by Eqs.~\eqref{r_jstar} and~\eqref{rstarg}.} let $Z$ and $Z'$ be defined as above, and let $\mathcal{A}(\vec{r})$ and $\mathcal{C}_{\mathrm{gate}}(\vec{r})$ be defined as in Eqs.~\eqref{A_app} and~\eqref{Cgate_app}. Then, for any $\gamma > 0$, we have
$|B| = |\E[Z] - \E[Z']| \leq \gamma$ if $M$ is any integer satisfying
\[ M \geq \frac{\ln(1/\gamma')}{W(\ln(1/\gamma')/e)}, \quad \text{with } \gamma' \coloneqq \frac{2\gamma}{\mathcal{A}(\vec{r})\mathcal{C}_{\mathrm{gate}}(\vec{r})}, \]
where $W(\cdot)$ denotes the principal branch of the Lambert-W function.
\end{theorem}
\begin{proof}
Using $|\tr[\rho A]| \leq \|A\|$ for any state $\rho$ and operator $A$, we have
\begin{align*}
|B| = \left|\E[Z] - \E[Z']\right| 
&\leq \sum_j |F_j| \Bigg\|\Bigg(\sum_{\substack{n=0\\n \text{ even}}}^\infty \alpha_n^{(j)} A_n^{(j)}\Bigg)^{r_j} - \Bigg(\sum_{\substack{n=0\\\text{$n$ even}}}^M \alpha_n^{(j)} A_n^{(j)}\Bigg)^{r_j}\Bigg\| \\
&\leq \sum_j |F_j| r_j \Bigg(\sum_{\substack{n=0\\\text{$n$ even}}}^\infty \alpha_n^{(j)} \|A_n^{(j)}\| \Bigg)^{r_j - 1} \Bigg(\sum_{\substack{n=M+1 \\n \text{ even}}}^\infty \alpha_n^{(j)} \|A_n^{(j)}\|\Bigg) \\
&\leq \sum_j |F_j| r_j \Bigg(\sum_{\substack{n=0\\\text{$n$ even}}}^\infty \alpha_n^{(j)} \Bigg)^{r_j-1} \sum_{\substack{n = M+1\\ n\text{ even}}}^\infty \alpha_n^{(j)} \\
&\leq \sum_j |F_j| r_j \Bigg(\sum_{\substack{n=0\\\text{$n$ even}}}^\infty \alpha_n^{(j)} \Bigg)^{r_j} \sum_{\substack{n = M+1\\ n\text{ even}}}^\infty \alpha_n^{(j)} \\
&\leq \sum_j  |F_j| r_j \mu_j \sum_{n=M+1}^\infty \frac{1}{n!} \Bigg(\frac{|t_j|}{r_j}\Bigg)^n.
\end{align*}
Here, the second inequality follows from the fact that 
\[ \|C^r - D^r\| \leq \sum_{k=1}^r\|C\|^{k-1}\|C-D\|\|D\|^{r-k} \leq r\max\{\|C\|,\|D\|\}^{r-1}\|C - D\| \]
for any operators $C$ and $D$.
The third inequality uses the fact that $\|A_n^{(j)}\| \leq 1$ since each $A_n^{(j)}$ is a convex combination of unitaries, and the fourth uses the observation that $\alpha_0^{(j)} \geq 1$ for all $j$. To obtain the last inequality, we use Eq.~\eqref{mu_j2} and the inequality $\sqrt{1 + a^2} \leq 1 + |a|$ for all $ a \in \mathbb{R}$, which implies
\[ \alpha_n^{(j)} \leq \frac{1}{n!} \left( \frac{t_j}{r_j}\right)^n \left(1 + \frac{1}{n+1} \frac{|t_j|}{r_j} \right). \]

Now, we use the assumption $|t_j| \leq r_j$ and Eq.~\eqref{poissontail} to bound
\[ \sum_{n=M+1}^\infty \frac{1}{n!}\left( \frac{|t_j|}{r_j}\right)^n \leq \frac{1}{2}\left(\frac{e}{M}\right)^{M}, \]
which is at most $\gamma'/2$ if $M > \frac{\ln(1/\gamma')}{W(\ln(1/\gamma')/e)}$. We then have $|B| \leq \sum_j |F_j|r_j\mu_j \gamma'/2 = \mathcal{A}(\vec{r})\mathcal{C}_{\mathrm{gate}}(\vec{r})\gamma'/2$, so the result follows by setting $\gamma' = 2\gamma/(\mathcal{A}(\vec{r})\mathcal{C}_{\mathrm{gate}}(\vec{r}))$.
\end{proof}

Thus, Theorem~\ref{thm:truncate} also shows that the truncation order $M$, which ultimately determines the classical sampling complexity of our algorithm, scales only logarithmically with the total quantum complexity, proportional to $\mathcal{A}(\vec{r})\mathcal{C}_{\mathrm{gate}}(\vec{r})$.

\section{Compiling to standard gates} \label{sec:GateComplexities}

In the main text, we counted the number $\mathcal{C}_{\mathrm{gate}}(\vec{r})$ of controlled Pauli rotations per sample,\footnote{In the main text, we described $\mathcal{C}_{\mathrm{gate}}$ as the number of (controlled) single-qubit Pauli rotations, as each multi-qubit Pauli rotation can be synthesised by conjugating a single-qubit $Z$ rotation by Clifford gates. Here, we present a different compilation strategy starting from the multi-qubit rotations.} ignoring the less important Clifford gate costs. This can be further compiled to more primitive gate sets, and we present rough estimates for doing so in this section.

First, each controlled Pauli rotation $\exp(i\theta P)$ can be decomposed into Cliffords and two single-qubit $Z$ rotations as
\begin{align*}
   \vert 1 \rangle \langle 1\vert  \otimes \exp(i \theta P) +  \vert 0 \rangle \langle 0 \vert  \otimes \id & =  \exp(i \theta \vert 1 \rangle \langle 1\vert \otimes P) \\
   & = \exp(i (\theta/2)(\id - Z_0 ) \otimes P) \\
   & = \exp(i (\theta/2)\id \otimes P ))  \exp(- i (\theta/2)Z_0 \otimes P ),
\end{align*}
where $Z_0$ is a Pauli $Z$ acting on the control qubit. Since $\{ \id \otimes P, - Z_0 \otimes  P \}$ is a commuting and independent set of Pauli operators, there exists a Clifford $C$ such that under conjugation by $C$, we have $\{ \id \otimes P, - Z_0 \otimes  P \} \rightarrow \{ Z_1, Z_2\}$. Thus, we find
\begin{align}
   \vert 1 \rangle \langle 1\vert  \otimes \exp(i \theta P) +  \vert 0 \rangle \langle 0 \vert  \otimes \id &   = C \exp(i (\theta/2) Z_1 ))  \exp(i (\theta/2)Z_2)) C^\dagger
\end{align}
leading to an extra factor of $\times2$ in non-Clifford complexity.

Second, each single-qubit $Z$ rotation of the form $\exp(i (\theta/2) Z_j ))$ can be compiled into the Clifford+$T$ gate set~\cite{bocharov15b,kliuchnikov13,gosset14,bocharov15,RS14}, though this typically increases gate counts by a factor $\mathcal{O}(\log(1/\epsilon))$ to achieve synthesis precision $\epsilon$. For instance, using the Ross-Selinger synthesis algorithm leads to an  synthesis overhead of $3\log_2(1/\epsilon) + O(\log(\log(1/\epsilon)))$.  For $\epsilon=10^{-10}$ this gives a $\sim\! 100\times$ overhead, which can be reduced to $\sim\!50\times$ using random compilation of the Ross-Selinger algorithm~\cite{campbell2017shorter}. Therefore, the expected $T$-count per sample is upper bounded by $\sim\! 100 \mathcal{C}_{\mathrm{gate}}(\vec{r})$.  

However, this large $100\times$ constant factor can be reduced by using smarter compilation strategies.  In our algorithm, the vast majority of gates are sampled from the leading order terms in Eq.~\eqref{eq:LeadingOrder}, which all have the same rotation angle. Let us assume that in our algorithm we have a subset of controlled Pauli rotations $\{ P_1, P_2, P_3, \ldots , P_w \}$ by the same angle. Furthermore, for chemistry problems these will typically be independent---there are no combinations that multiply to form the identity---and we assume they are all commuting within this subset (we address validity of this assumption later).  Then, the set of $w$ controlled Pauli rotations becomes a sequence of Pauli rotations with respect to the set
\begin{equation}
\{ \id \otimes P_1 , - Z_0 \otimes P_1 , \id \otimes P_2 , - Z_0 \otimes P_2 , \ldots , \id \otimes P_w , - Z_0 \otimes P_w \},
\end{equation}
which is a set of $2w$ commuting and independent Pauli operators.  For such a set of operators, there will exist a Clifford rotation $C$ that maps this set to $\{ Z_1, Z_2, \ldots , Z_{2w} \}$.  Therefore, the $w$ controlled-Pauli rotations can be realised by $C \prod_{j=1}^{2w} \exp(i (\theta/2) Z_j) C^\dagger$. The special structure of $\prod_{j=1}^{2w} \exp(i (\theta/2) Z_j)$  allows the use of Hamming weight phasing~\cite{gidney2018halving,kivlichan2018quantum,campbell2020early}. Notice that
\begin{equation}
    \exp(i (\theta/2) Z_j) \ket{x} =  \exp(i \theta ( w- |x| )) \ket{x},
\end{equation}
where $|x|$ is the Hamming weight of bit-string $x$ over the indices 1 to $2w$.  Therefore, the key idea of Hamming weight phasing is that we produce the phase $\exp(i \theta ( w- |x| ))$ by first mapping $\ket{x} \rightarrow \ket{x}\ket{|x|}$ where $\ket{|x|}$ is a register of size $\log_2(2w)$ qubits storing a binary representation of the integer $|x|$.  We now need only $\log_2(2w)$ Pauli $Z$ rotations acting on the $\ket{|x|}$ register. The cost of calculating the Hamming weight on $2w$ bit is upper bounded by $2w$ Toffoli gates (and ancilla qubits). For example, costing each Pauli rotation at 50 $T$ gates, equivalently 25 Toffoli gates using catalysis, the total Toffoli cost is then
\begin{equation}
    \mathcal{C}_{w-\mathrm{Rot}} = w \left( 2w + 25 \log_2( 2w )   \right) .
\end{equation}
We plot the Toffoli cost per gate $\mathcal{C}_{w-\mathrm{Rot}}/w$ in \cref{fig:HWP}.  For large $w$, the Toffoli cost per gate will approach 2.  We see that for finite $w$, at $w=100$ we need $\sim\! 4$ Toffoli per gate and at $w=40$ we need $\sim\! 6$ Toffoli per gate.

\begin{figure}
    \centering
    \includegraphics[width=300pt]{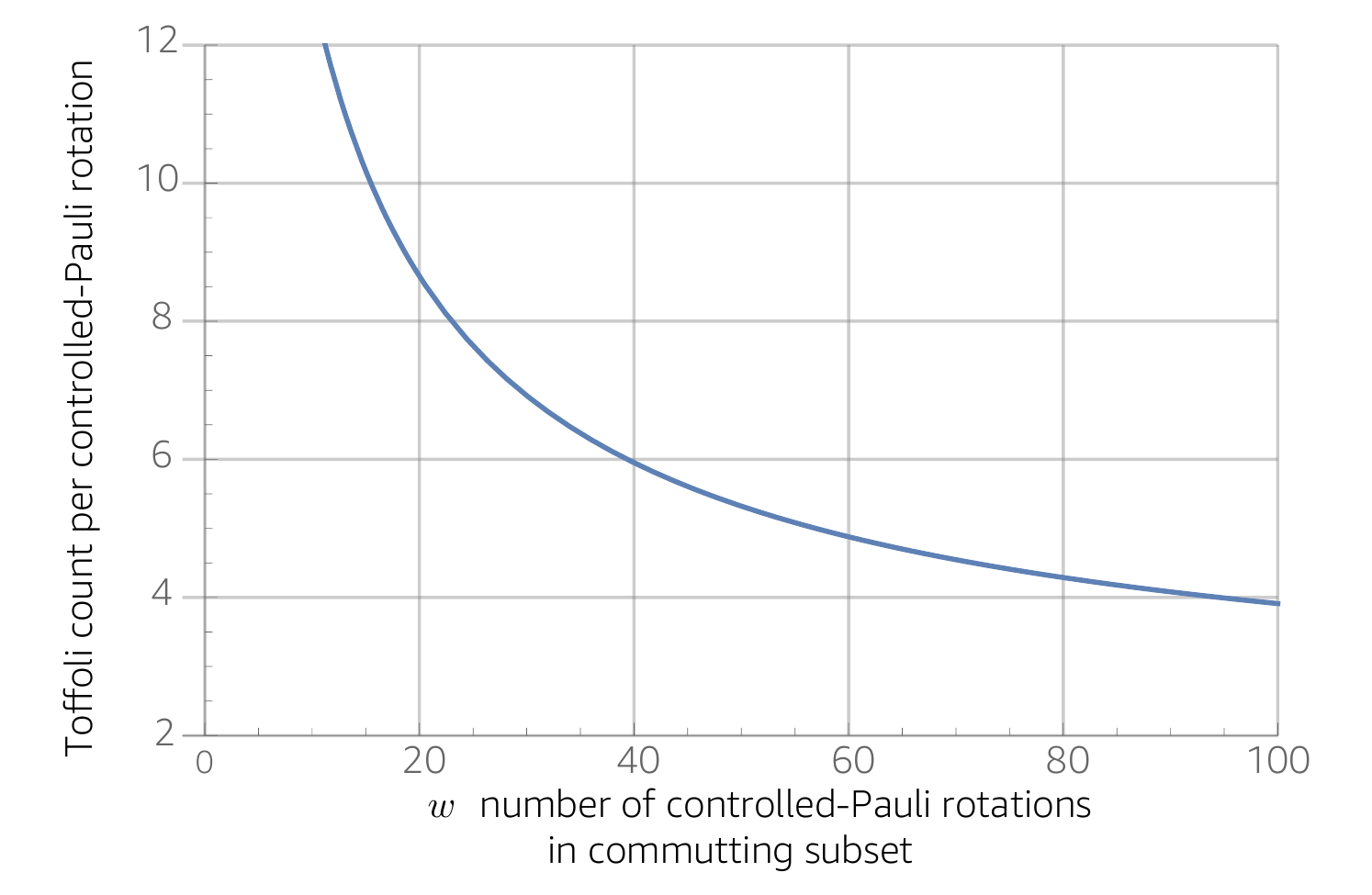}
    \caption{Upper-bound on the Toffoli cost of implementing $w$ controlled-$\exp(i \theta P_j)$ rotations, where all $P_j$ are commuting and independent. We assume that each Pauli $Z$ rotation require 25 Toffoli gates, using catalysis and randomized Ross-Selinger synthesis \cite{campbell2017shorter}.}
    \label{fig:HWP}
\end{figure}

Assuming the randomly sampled unitaries are dominated by long commuting sequences of the form $\{ P_1, P_2, \ldots , P_w \}$ with large $w$ (e.g. $w\gg 100$) and identical rotation angles, this supports the claim in the main text that, roughly, the Toffoli cost scales as $2 \mathcal{C}_{\mathrm{gate}}$. But using a more modest number of ancilla (e.g. $w \sim 40$), the Toffoli cost scales as $6 \mathcal{C}_{\mathrm{gate}}$.



Of course, there is a small but nonzero probability that we sample higher order terms ($n>0$) with rotation angles $\theta_n \neq \theta_0$ (cf.~\cref{eq:angles}). we can perform these controlled rotations using standard---though expensive---circuit synthesis instead of Hamming weight phasing. However, the frequency of these rotations is far fewer than one in every 100 gates, so this extra expense is relatively negligible. There is also a finite probability that a Pauli rotation $\exp(i\theta P_j)$, is followed by a sample $P_k$ that does not commute with $P_j$. However, there are $\mathcal{O}(N^4)$ Pauli operators in the Hamiltonian, and for any given $P_j$ there are only $\mathcal{O}(N^3)$ non-commuting Pauli terms. With each Pauli equally weighted, the probability of selecting an anti-commuting operator is $\mathcal{O}(1/N)$. Therefore, for large enough $N$ we expect to encounter many long sequences of commuting rotations that enable the use of Hamming weight phasing.  For these reasons, we stress that the $2 \mathcal{C}_{\mathrm{gate}} - 6 \mathcal{C}_{\mathrm{gate}}$ Toffoli count claims are a rough, asymptotic estimate. A more detailed analysis of the finite-$N$ statistics and pre-asymptotics is beyond the scope of this work.

\section{Alternative rescaling factors}
\label{App:TAU}

When we defined the CDF $C(\cdot)$ in Eq.~\eqref{cdfdef}, we rescaled the Hamiltonian by a factor $\tau$. Our analysis in the main text proceeds on the assumption that $\tau$ is set to $\pi /(2 \lambda + \Delta)$. Under this assumption, our results follow in a fully rigorous manner. However, reduced resource overheads can be obtained via heuristic modifications of the value used for $\tau$.  

The CDF is inferred from expectations of $\mathrm{Tr}[\rho e^{i j H \tau}]$, and so to avoid ambiguity due to periodicity of this exponential Proposition~\ref{prop:guarantee} required that $\tau$ was chosen small enough that $p(x)$ is supported within the interval $x \in [- (\pi - \delta)/2 , (\pi - \delta)/2 ]$.  Recall that if the state $\rho$ is supported on eigenstates with eigenvalues in the range $[-E, +E]$ for some $E$, then $p(x)$ is supported on $x \in [- \tau E, \tau E]$.  It follows that Proposition~\ref{prop:guarantee} can be employed whenever
\begin{equation}
    \tau E \leq (\pi - \delta)/2 .
\end{equation}
Using $\delta =\tau \Delta$ and simplifying, this equates to
\begin{equation}
    \tau \leq \frac{\pi}{2E + \Delta}
\end{equation}
In Appendix~\ref{app:acdf} and throughout the main text, we used that $E \leq \| H \| \leq \lambda$. However, this analysis is overly pessimistic and we could often set  
\begin{equation} \label{eq:bDEF}
    \tau =  \frac{\pi}{2 \lambda/b + \Delta}.
\end{equation}
for some $b>1$ without any significant problems as we explain below.

First, $\| H \| \leq \lambda$ is typically very loose for frustrated systems.  If we know $\|H\|$, then we can determine a new range of safe values for $\tau$ and therefore $b$.  However, calculating $\|H\|$ is computationally hard and so typically its value is unknown; hence, the use of $b=1$ for our rigorous theorem statements.

\begin{figure}
    \centering
    \includegraphics[width=250pt]{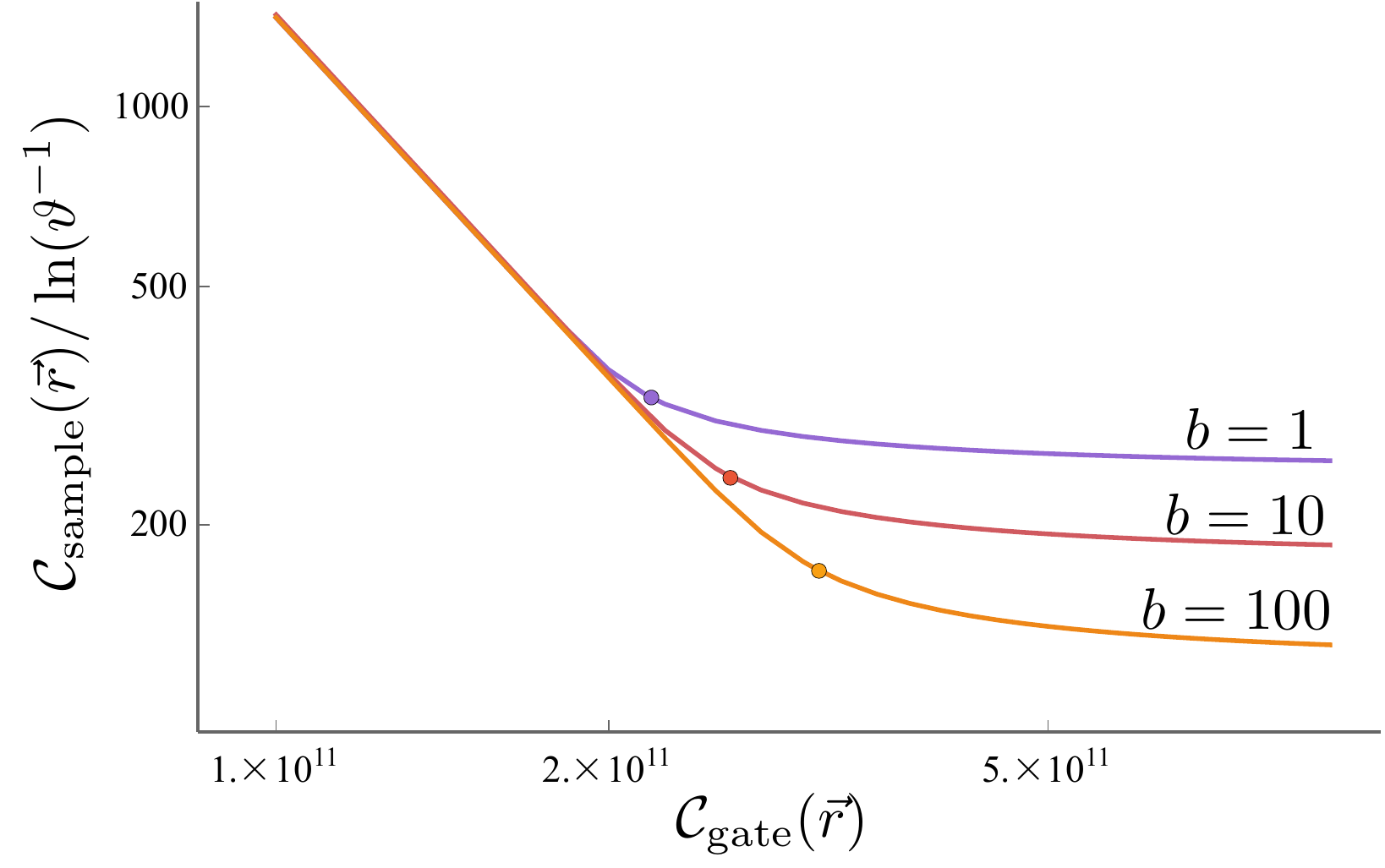}
    \caption{Similar to Fig.~\ref{fig:FeMoco}, this is a log-log plot of $\mathcal{C}_{\mathrm{sample}}(\vec{r})/\ln(\vartheta^{-1})$ vs.~$\mathcal{C}_{\mathrm{gate}}(\vec{r})$, for $\lambda = 1511$ (FeMoco~\cite{LiFeMoco,koridon2021orbital}) with $\Delta= 0.0016$ (chemical accuracy), $\eta = 1$, and various runtime vectors $\vec{r}$ optimised using the algorithms in Appendix~\ref{app:optimiser}. Whereas Fig.~\ref{fig:FeMoco} explored different $\eps$ values, here we fix $\eps=0.2$ and vary $b \in \{1,10,100 \}$ as defined by Eq.~\eqref{eq:bDEF}, showcasing the resource improvements possible with $b>1$. Dots indicate the values that optimise the total expected complexity $2\mathcal{C}_{\mathrm{sample}}\cdot \mathcal{C}_{\mathrm{gate}}$, while curves are obtained by fixing $\mathcal{C}_{\mathrm{gate}}$ and optimising $\mathcal{C}_{\mathrm{sample}}$. 
    }
    \label{fig:varyB}
\end{figure}

Second, the assumptions of Proposition~\ref{prop:guarantee} can be relaxed with a similar proof going through.  That is, let us assume that $p(x)$ is \textit{mostly} supported on the interval $x \in [- (\pi - \delta)/2 , (\pi - \delta)/2 ]$, so that the support outside this interval has total weight no more than $\epsilon'$.  Then one could derive a similar result to Proposition \ref{prop:guarantee} at the price of an extra $\epsilon'$ to the additive error bounds on the CDF. Provided $\epsilon'$ is small compared to $\eta$, this extra error could be accommodated by a slight tuning of the algorithm parameters.  When will this assumption on $p(x)$ hold?  The initial state $\rho$ is typically taken to be an approximation of the ground state, so it will have very low energy $\mathrm{Tr}[\rho H]$, close to the ground state energy. Indeed, the ground state energy and also $\mathrm{Tr}[\rho H]$ could be several orders of magnitude smaller than $\|H\|$.  When $\mathrm{Tr}[\rho H] \ll \|H\|$, $\rho$ cannot have large overlap with high-energy eigenstates.  Thus, in practice, it will often be safe to set $b$ such that $ b > \lambda / \|H\|$ (but not very much larger) since then any support of $p(x)$ outside $[- (\pi - \delta)/2 , (\pi - \delta)/2 ]$ will be relatively small.


\end{document}